\documentclass[lettersize,journal]{IEEEtran}

\usepackage{amssymb,amsmath,amsthm}
\usepackage[lined,boxed,commentsnumbered,ruled]{algorithm2e}
\usepackage{algorithmic} %format of the algorithm
\usepackage{graphicx}
\usepackage{stmaryrd}
\SetSymbolFont{stmry}{bold}{U}{stmry}{m}{n}
\usepackage{multirow,url}
\usepackage{color,cite}
\usepackage{cite}
\usepackage{tikz,pgf}
\usetikzlibrary{fadings,shapes,arrows,shadows}
\usepackage{makecell}
\usepackage{array}
\usepackage{overpic} %text pictures
\usepackage{subfigure}
\usepackage{url}
\usepackage{float}
\usepackage{booktabs} % For formal tables
\usepackage{bm}
\usepackage{upgreek}
\usepackage{mathrsfs}
\usepackage{verbatim}

\newtheorem{lemma}{Lemma}

\newtheorem{definition}{Definition}
\newtheorem{theorem}{Theorem}

\theoremstyle{plain}

\ifodd 1
\newcommand{\rev}[1]{{#1}} %revise of the text
\newcommand{\com}[1]{\textbf{\color{red} (COMMENT: #1)}} %comment of the text
\newcommand{\comg}[1]{\textbf{\color{green} (COMMENT: #1)}} %comment of the text
\newcommand{\response}[1]{\textbf{\color{green} (RESPONSE: #1)}} %response to comment
\else
\newcommand{\rev}[1]{#1}

\newcommand{\com}[1]{}
\newcommand{\comg}[1]{}
\newcommand{\response}[1]{}
\fi
\def\N{\mathcal{N}}

\IEEEoverridecommandlockouts

\begin{document}

%\title{ Collaborative Block Mining and Edge Task Offloading in MEC-Assisted Blockchain Networks: A Stackelberg Coalition Game-Theoretic Approach } %Conference title

\title{ An Overlapping Coalition Game Approach for Collaborative Block Mining and Edge Task Offloading in MEC-assisted Blockchain Networks }

\author{Licheng Ye, \emph{Graduate Student Member, IEEE}, Zehui Xiong, \emph{Senior Member, IEEE}, Lin Gao, \emph{Senior Member, IEEE}, and Dusit Niyato, \emph{Fellow, IEEE}
\thanks{This work was supported in part by the National Key Research and Development Program of China under Grant 2021YFB2900300, in part by the Natural Science Foundation of Guangdong Province under Grant 2024A1515010178, in part by the Shenzhen Science and Technology Program under Grant KQTD20190929172545139, Grant GXWD20231129103946001, Grant KJZD20240903095402004, and Grant ZDSYS20210623091808025, in part by the Guangdong Basic and Applied Basic Research Foundation under Grant 2023A1515012819, and in part by the National Research Foundation, Singapore and Infocomm Media Development Authority under its Future Communications Research \& Development Programme (FCP-NTU-RG-2022-010 and FCP-ASTAR-TG-2022-003), Singapore Ministry of Education (MOE) Tier 1 (RG87/22 and RG24/24), the NTU Centre for Computational Technologies in Finance (NTU-CCTF), and the RIE2025 Industry Alignment Fund - Industry Collaboration Projects (IAF-ICP) (Award I2301E0026), administered by A*STAR, as well as supported by Alibaba Group and NTU Singapore through Alibaba-NTU Global e-Sustainability CorpLab (ANGEL).}
\thanks{L. Ye and L. Gao are with the School of Electronics and Information Engineering and the Guangdong Provincial Key Laboratory of Aerospace Communication and Networking Technology, Harbin Institute of Technology, Shenzhen 518055, China (e-mail: gavin\_ye.km@foxmail.com; gaol@hit.edu.cn).
Z. Xiong is with the School of Electronics, Electrical Engineering and Computer Science (EEECS), Queen's University Belfast, Belfast, BT7 1NN, U.K. (z.xiong@qub.ac.uk).
D. Niyato is with the College of Computing and Data Science, Nanyang Technological University, Singapore (email: dniyato@ntu.edu.sg).
(\emph{Corresponding Author: Lin Gao})
}
\thanks{Part of the results have been published in IEEE GLOBECOM 2023 \cite{add-1}.}
%Lin Gao is the Corresponding Author. Email: gaol@hit.edu.cn.
%}
%%\thanks{This work is supported by the National Natural Science Foundation of China  (Grant No. 61771162 and 61501211) and the Basic Research Project of Shenzhen (Grant No. JCYJ20160531192013063 and JCYJ20170307151148585).}
}

\IEEEtitleabstractindextext{
\begin{abstract}
% MECµÄ±³¾°

Mobile edge computing (MEC) is a promising technology that enhances  the efficiency  of mobile blockchain networks, by enabling miners, often acted by mobile users (MUs) with limited computing resources, to offload resource-intensive mining tasks to nearby edge computing servers.
Collaborative block mining can further boost  mining efficiency by allowing multiple miners to form \emph{coalitions}, pooling their computing resources and transaction data together to mine new blocks collaboratively.
Therefore, an MEC-assisted collaborative blockchain network can leverage the strengths of both technologies, offering improved efficiency, security, and scalability for blockchain systems.
While existing research in this area has mainly focused on the \emph{single-coalition} collaboration mode, where each miner can only join one coalition, this work explores a more comprehensive \emph{multi-coalition} collaboration mode, which allows each miner to join multiple coalitions.
To analyze the behavior of miners and the edge computing service provider (ECP) in this scenario, we propose a novel two-stage Stackelberg game.
%, which consists of a resource pricing problem (for the ECP) in Stage I, and a coalition formation game (for the miners) and a resource competition game (for the coalitions) in Stage II. Specifically,
In Stage I, the ECP, as the leader, determines the prices of computing resources for all MUs.
In Stage II, each MU decides the coalitions to join, resulting in an \emph{overlapping coalition formation (OCF) game}; Subsequently, each coalition decides how many edge computing resources to purchase from the ECP, leading to an \emph{edge resource competition (ERC) game}.
We derive the closed-form Nash equilibrium for the ERC game, based on which we further propose an OCF-based alternating algorithm to achieve a stable coalition structure for the OCF game and develop a near-optimal pricing strategy for the ECP's resource pricing problem.
Simulation results show that the proposed multi-coalition collaboration mode can improve the system efficiency by $12.64\% \sim 17.63\%$, compared to the traditional single-coalition collaboration mode.

%\emph{Mobile edge computing (MEC)} is a promising technology for improving the efficiency and security of mobile blockchain networks, by allowing miners with limited computing resources to offload the computation-intensive mining tasks to edge computing servers that are proximate to them. \emph{Collaborative block mining} can further improve the mining efficiency and increase the miner profit, by enabling multiple miners to pool their computing resources and transaction data together to collaboratively mine new blocks.
%a \emph{coalition formation game} as Stage II.A, and an \emph{edge resource competition game} (among the formed coalitions) as Stage II.B.
%Specifically, in Stage I, the ECP acts as a leader and decides the price of computing resources for all MUs. In Stage II.A, each miner acts as a game player and selects multiple coalitions to join, leading to an OCF game among miners. In Stage II.B, each formed coalition acts as a game player and decides the amount of edge computing resources to invest, leading to an ERC game among coalitions.
%We derive the closed-form Nash equilibrium for the ERC game, and propose an OCF-based alternating algorithm that converges to a stable coalition structure for the OCF game and a near-optimal pricing strategy for resource pricing.
\end{abstract}
\begin{IEEEkeywords}
Blockchain, collaborative block mining, mobile edge computing, coalition game.
\end{IEEEkeywords}}

\maketitle
\IEEEdisplaynontitleabstractindextext
\IEEEpeerreviewmaketitle

%%%%%%%%%%%%%%%%%%%%%%%%%%%%%%%%%%%%%%%%%%%%%%%%%%%%%%%% INTRODUCTION %%%%%%%%%%%%%%%%%%%%%%%%%%%%%%%%%%%%%%%%%%%%%%%%%%%%%
\section{Introduction}\label{Introduction}

\subsection{Background and Motivations}\label{Backgroud}

Blockchain is an innovative decentralized ledger technology that ensures secure and transparent transactions without the need for a centralized intermediary \cite{1,2,3}.
%It achieves this through cryptography and consensus mechanisms, safeguarding data integrity and preventing unauthorized tampering.
It achieves this through the use of cryptography and consensus mechanisms, which protect data integrity and prevent unauthorized tampering.
The decentralized and distributed nature of blockchain enables secure and trustless communication among mobile devices in wireless networks.
This helps reduce communication latency and enhances network scalability.
As a result, it has been considered as a key enabling technology for the future 6G network \cite{4,5}.

In blockchain systems, transactions are recorded in \emph{blocks} and sequentially linked to form a chain, called ``blockchain''.
In this system, \emph{miners} play a crucial role in validating transactions and appending them to the blockchain through a \emph{mining process}.
During this process, miners compete to publish blocks by solving proof-of-work (PoW) puzzles, which requires substantial computational resources.
In mobile blockchain networks, miners are typically mobile users (MUs) using portable devices (e.g., smartphones and laptops).
These devices may not have sufficient computing resources to execute the computation-intensive mining tasks \cite{7,8,9}.
To this end, \emph{Mobile edge computing (MEC)} \cite{10,11,12,13} has emerged as a promising solution to enhance the efficiency of mobile blockchain networks.
It allows miners with limited computing resources to offload the computation-intensive mining tasks to edge computing servers that offer greater processing  capabilities.
Existing studies on MEC-assisted blockchain networks \cite{14,15,16,17,add-3,add-4} often assume that miners act in a \emph{non-cooperative} manner, where miners mines new blocks independently and competitively.
This approach, however, can result in resource wastage and profit loss for miners.

To enhance mining efficiency and increase miner profitability, some researchers have proposed a novel \emph{collaborative block mining} scheme \cite{18,19, 20}.
This scheme enables multiple miners to form a \emph{coalition}, pooling their computing resources and transaction data together to mine new blocks collaboratively.\footnote{Typical examples of such coalitions are  \emph{mining pools} (e.g., F2Pool, Antpool, and ViaBTC) for commercial blockchain networks like Bitcoin   \cite{19}.}
First, pooling computing resources allows miners to work together on mining tasks,   reducing resource redundancy and improving  mining efficiency.
Second, aggregating transaction data enables miners to prioritize transactions with higher fees when constructing new blocks, thereby maximizing the expected rewards \cite{20}.
As a result, an \emph{MEC-assisted collaborative blockchain network} that integrates both MEC and collaborative block mining can leverage the combined advantages of these technologies to improve efficiency, security, and scalability of blockchain systems, ultimately promoting the widespread adoption of blockchain technology.
Existing research in this area \cite{21,22,23,24, add-5} mainly focused on the \emph{single-coalition} collaboration mode, where each miner participates in only one collaborative group (\emph{coalition}).
However, this limitation restricts the full potential of collaborative block mining.
In this work, we   explore a more comprehensive \emph{multi-coalition} collaboration mode, where each miner can join multiple coalitions,   leading to an \emph{Overlapping Coalition Formation (OCF) game}.\footnote{Here, the term ``overlapping'' refers to the situation  where different coalitions   share some common miners.
The OCF game has been applied in various wireless networks to model user collaboration, including collaborative smartphone sensing \cite{smartphone-sensing}, non-orthogonal multiple access systems \cite{25}, and collaborative unmanned aerial vehicle networks \cite{26}.}

%While the OCF game has previously found application in wireless networks, including non-orthogonal multiple access systems \cite{25} and heterogeneous unmanned aerial vehicle networks \cite{26}, this is the pioneering work that uses the OCF game to analyze an MEC-assisted blockchain network with multi-coalition collaboration mode.
%In \cite{8}, the cooperative game model is used to study the dynamics of mining pools and how pool members share the mining rewards.
%In wireless blockchain networks, \cite{9} studies the MEC-aid computing resource allocation problem through the perspective of collaboration, and enhances the system utility based on coalition formation game. However, the above works assume that each miner can only join one coalition, which is the single-coalition collaboration mode. Such assumptions simplify the collaborative block mining for miners.

\begin{figure}[t]
	\centering
	\includegraphics[width=3.1in]{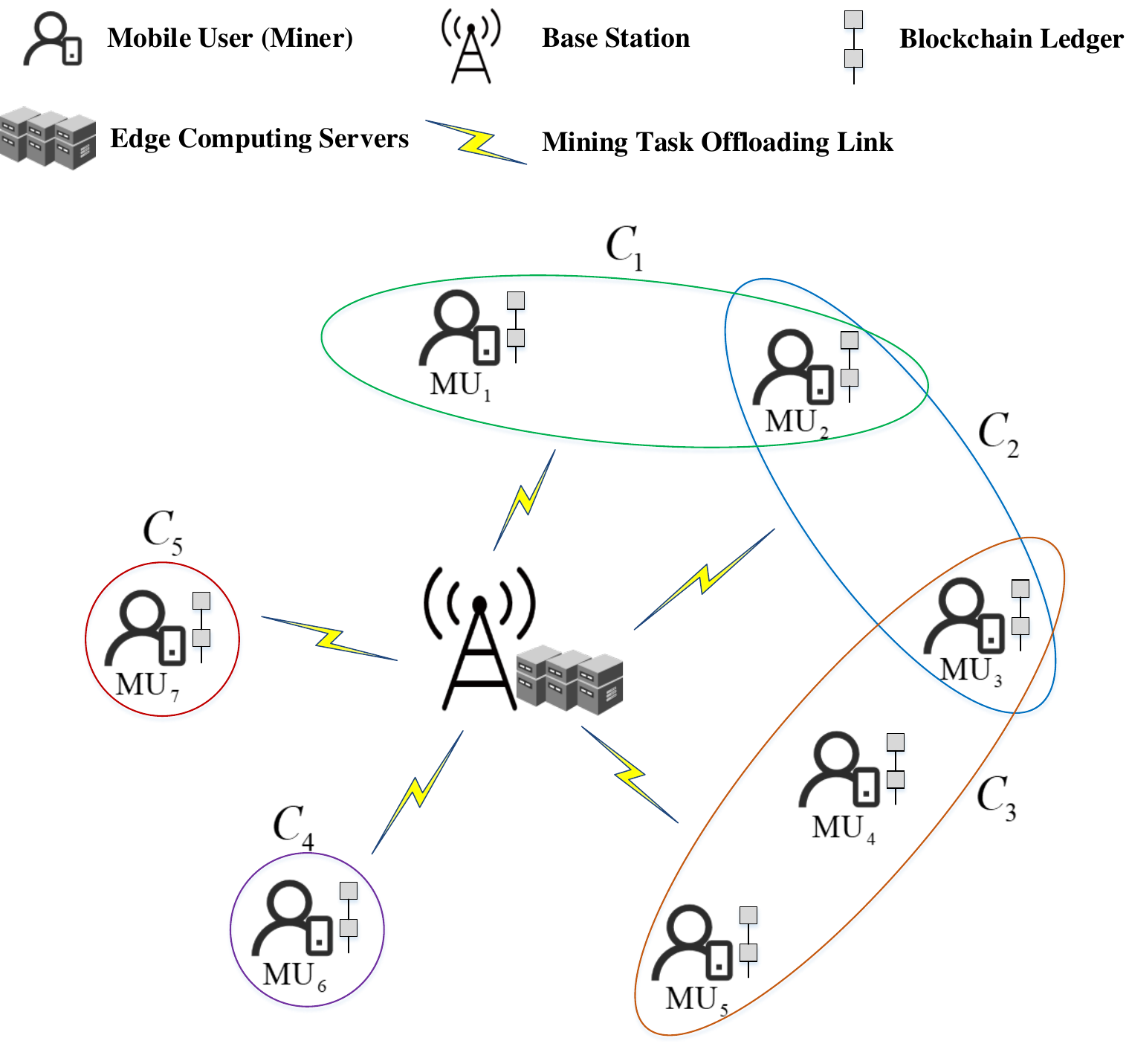}
	\vspace{-3mm}
	\caption{An Example of MEC-assisted Collaborative Blockchain Network with Multi-Coalition Collaboration Mode.}
	\label{fig:system-model}
	\vspace{-3mm}
\end{figure}

\subsection{Solution and Contributions}

Specifically, in this work, we investigate an MEC-assisted collaborative blockchain network, which consists of multiple MUs with limited computing resources, acting as miners and mining blocks collaboratively, and one edge computing service provider (ECP), deploying edge computing servers on base stations for MUs to offload their mining tasks.
Moreover, we explore a novel multi-coalition collaboration mode, where MUs can form coalitions to mine blocks collaboratively, and meanwhile each MU can join multiple coalitions.
Compared to the traditional single-coalition collaboration mode in \cite{21,22,23,24,add-5}, this new mode can improve the potential of collaborative block mining.
It also more accurately reflects the behavior of miners in real-world blockchain networks, where miners are often not restricted to joining a single mining pool.
Fig.~\ref{fig:system-model} illustrates such an MEC-assisted collaborative blockchain network with $7$ MUs and $5$ coalitions, where $\rm{MU}_{1}$ and $\rm{MU}_{2}$ form a coalition $C_1$,
$\rm{MU}_{2}$ and $\rm{MU}_{3}$ form a coalition $C_2$,
$\rm{MU}_{3}$, $\rm{MU}_{4}$, and $\rm{MU}_{5}$ form a coalition $C_3$,
and
$\rm{MU}_{6}$ and $\rm{MU}_{7}$ form the single-user coalitions $C_4$ and $C_5$, respectively.
It is clear to see that $C_1$ and $C_2$ overlap as they share the common member $\rm{MU}_{2}$.
Similarly, $C_2$ and $C_3$ overlap as they share the common member $\rm{MU}_{3}$.

In this scenario, we aim to explore the following problems for MUs, coalitions, and the ECP, respectively.

\emph{1) For MUs:} %Whether and how to form coalitions for collaborative block mining?
Should they form coalitions for collaborative block mining, and if so, how should they do so?

\emph{2) For the formed coalitions:} How much computing resources should they purchase from the ECP?

\emph{3) For the ECP:} How should it price the computing resources offered to  MUs?

To analyze the behavior of miners, coalitions, and the ECP in such a scenario, we formulate a novel \emph{two-stage Stackelberg game}, as shown in Fig.~\ref{Framework}, which consists of a resource pricing problem (for the ECP) in Stage I, followed by a   coalition formation game (for the miners) and a resource competition game (for the coalitions) in Stage II\footnote{Stackelberg game has been widely used to analyze user behaviors and strategic interactions in wireless networks \cite{game-1,game-2,game-3}.}.
% presents the framework of the proposed game.
More specifically, in Stage I, the ECP, acting as the leader, determines the prices of computing resources for all MUs.
In Stage II, each MU first selects one or multiple coalitions to join, resulting in an \emph{Overlapping Coalition Formation (OCF) game}; and subsequently, each coalition decides how much edge computing resources to purchase from the ECP, leading to an \emph{Edge Resource Competition (ERC) game}.\footnote{In practice, the coalition decision can be made by the head miner in each coalition, e.g., the mining pool manager.}
%Note that the proposed game is very challenging.
%On one hand, the ``overlapping'' case increases the analytical complexity in the coalition game.
%On the other hand, the coupling of the OCF game with the ERC game further enhances the difficulty of the analysis.
It is important to note that the proposed game presents significant challenges, due to  the ``overlapping'' nature of coalitions as well as the
the coupling of the OCF game with the ERC game.
%On one hand, the ``overlapping'' nature of the coalitions adds analytical complexity to the game. On the other hand, the coupling of the OCF game with the ERC game further complicates the analysis.

We address these challenges and analyze the game equilibrium systematically using backward induction.
Specifically, for the ERC game among coalitions in Stage II.B, we derive the closed-form Nash equilibrium (NE) based on the Karush-Kuhn-Tucker (KKT) conditions.
For the OCF game in Stage II.A and resource pricing in Stage I, we propose an OCF-based alternating algorithm that converges to a stable coalition structure for the OCF game and provides a near-optimal pricing strategy for the ECP.
In summary, the main contributions of this work are as follows.

%For the OCF game among MUs in Stage II.A, we propose an OCF-based alternating algorithm that converges to a stable coalition structure.
%For the computing resource pricing problem of the ECP in Stage I, we employ a one-dimensional search to find the optimal price.
%The computational power allocation sub-problem in each coalition structure can be transformed into the unconstrained optimization form and solved by Newton's method. For the problem of MU grouping, we adopt an OCF game-based solution that allows MUs to form coalition structure in an overlapping manner. Furthermore, we propose the OCF-based joint computational power allocation and MU grouping optimization algorithm, which includes rules for coalition structure changes. The proposed algorithm can converge to a stable coalition structure and thus achieve the maximum of the blockchain system sum utility approximately. Simulation result shows that compared with non-overlapping strategy and non-cooperative strategy, the proposed OCF-based algorithm can obviously enhances the blockchain system sum utility. In summary, the main contributions of this paper are listed as follows

\begin{figure}[t]
	\centering
	\includegraphics[width=3.1in]{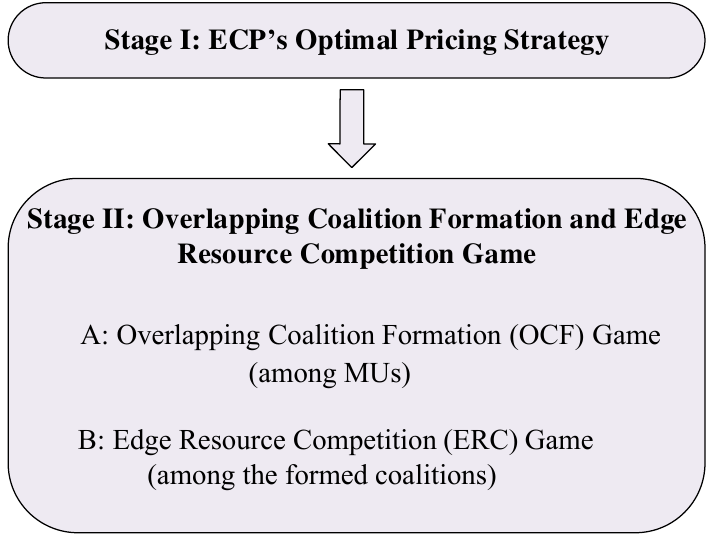}
	\caption{Framework of the Proposed Game.}
	\label{Framework}
	\vspace{-3mm}
\end{figure}

\begin{itemize}
	\item \textit{Novel Model Scenario}:
	We study a novel MEC-assisted collaborative blockchain network with a multi-coalition collaboration mode, which allows each MU to join multiple coalitions.
This new mode is more efficient, and meanwhile more accurately reflects the behavior of
miners in real-world blockchain networks.
%	This collaboration mode is more suitable for the practical blockchain networks, and can improve the potential of collaborative block mining.
	\item \textit{Comprehensive Game-Theoretic Analysis}:
	We formulate a novel two-stage Stackelberg game to characterize the collaborative and competitive behaviors of MUs.
	We provide a comprehensive analysis for the game equilibrium, and further design the OCF-based alternating algorithm that converges to the equilibrium.
	%consisting of resource pricing  for the ECP in Stage I, and an \emph{OCF and ERC game} for MUs and the formed coalitions in Stage II.

	\item \textit{Performance Evaluation and Insights}:
	Numerical simulation results show that the proposed multi-coalition collaboration mode can significantly enhance block mining efficiency.
	In particular, compared to the non-cooperative mode, the single-coalition collaboration mode can increase the system's utility by  $10.42\% \sim 12.48\%$, while the multi-coalition collaboration mode can further improve the system's utility by $12.64\% \sim 17.63\%$.

\end{itemize}

The rest of this work is organized as follows.
In Section \ref{Literature Review}, we provide the literature review.
In Section \ref{SystemModel}, we present the system model and the problem formulation.
In Section \ref{problemFormulation}, we analyze the two-stage Stackelberg game and propose the OCF-based alternating iteration algorithm.
We present the numerical simulation results in Section \ref{simulation}, and finally conclude in Section \ref{conclusion}.

%%%%%%%%%%%%%%%%%%%%%%%%%%%%%%%%%%%%%%%%%%%%%%%%%%%%%%%% LITERATURE REVIEW %%%%%%%%%%%%%%%%%%%%%%%%%%%%%%%%%%%%%%%%%%%%%%%%%%%%%
\section{Literature Review}\label{Literature Review}

\subsection{MEC-assisted Blockchain Network}

Mobile edge computing has garnered significant attention within blockchain networks due to its accessibility to computing resources.
Many recent studies have investigated MEC-assisted blockchain networks \cite{14,15,16,17}.
For example, Xiong \emph{et al.} \cite{14} explored the management of edge computing resources and pricing in blockchain networks.
Guo \emph{et al.} \cite{15} addressed the challenge of edge task offloading in mobile blockchains using Stackelberg game theory and double auction mechanisms.
Jiang \emph{et al.} \cite{16} introduced a multi-leader multi-follower Stackelberg game model for MEC-assisted block mining.
Wang \emph{et al.} \cite{17} examined the computation offloading problem for miners within MEC-assisted blockchain networks using a Markov game framework, proposing a learning-based algorithm to ensure Nash equilibrium.
\rev{Huang \emph{et al.} \cite{add-3} investigates the resource pricing and scheduling in edge-assisted blockchain networks using a three-stage multi-leader multi-follower Stackelberg game model, where each edge provider's pricing strategy is derived through best response algorithms.
Ling \emph{et al.} \cite{add-4} investigates the strategies of MEC servers in blockchain networks, where the MEC servers price the computing resources and decide whether they mine for themselves or not.}
These prior studies have utilized game theory or auction theory to examine the issue of mining task offloading.
However, they mainly focused on the \emph{non-cooperative} mining scheme, where   miners compete for the block reward, and only one miner can win the reward for each generated block, leading to significant wastage of computing resources among other miners.

To this end, some studies have explored the collaborative block mining scheme \cite{21,22,23,24}.
Collaborative block mining originates from the concept of 'mining pools' in blockchain networks.
It enables miners to aggregate computing resources in a pool (coalition) to collectively mine blocks, thereby enhancing their winning probability.
Lewenberg \emph{et al.} \cite{21} analyzed Bitcoin mining pool dynamics using cooperative game theory, where each miner can choose to join one mining pool.
Similarly, Zhao \emph{et al.} \cite{22} investigated the coalition formation game for miners in an MEC-assisted wireless collaborative blockchain network, addressing edge computing resource allocation, with each miner restricted to a single coalition.
Lajeunesse \emph{et al.} \cite{23} examined cooperative mining issues among mining pools and miners through a two-stage Stackelberg game.
Mai \emph{et al.} \cite{24} investigated the mining pool selection problem, proposing a centralized algorithm based on evolutionary game theory.
\rev{Additionally, Wang \emph{et al.} \cite{add-5} addresses the task offloading problem in mobile blockchain systems for IoT using a mixed model of Stackelberg and coalition formation games, aiming to optimize group partitioning and resource pricing between IoT miners and the computing service provider with low complexity.}
However, these studies assume a \emph{single-coalition} collaboration mode, where each miner can join only one coalition, thereby limiting the potential of collaborative block mining.
In this work, we explore a more comprehensive \emph{multi-coalition} collaboration mode, enabling miners to join multiple coalitions, thereby increasing the potential of collaborative block mining.

\subsection{Overlapping Coalition Game (OCF)}
The coalition game \cite{27} is a concept from game theory \cite{28} used to model cooperation among players.
Different from the traditional coalition game, players in the OCF game can belong to multiple coalitions simultaneously, referred to as the ``overlapping'' coalition structure.
The OCF games have found widespread  applications in wireless communication scenarios \cite{smartphone-sensing,25,26}.
For instance, Di \emph{et al.} \cite{smartphone-sensing} introduced an incentive mechanism based on the OCF game to reward smartphone users participating in sensing tasks.
Chen \emph{et al.} \cite{25} proposed an OCF game-based user grouping scheme aimed at maximizing the system sum rate in non-orthogonal multiple access systems.
Qi \emph{et al.} \cite{26} introduced a sequential OCF game to investigate the resource allocation problem in heterogeneous unmanned aerial vehicle networks.
Given the revenue advantages of the OCF game, we adopt this approach to model the multi-coalition collaboration mode in MEC-assisted blockchain networks.
It is important to note that the game proposed in this work presents significant challenges, due to the ``overlapping'' nature of coalitions as well as the the coupling of the OCF game with the ERC game.
Therefore, traditional methods in \cite{smartphone-sensing,25,26} cannot be directly applied to the proposed game in this work.

%On one hand, the ``overlapping'' nature of the coalitions adds analytical complexity to the game. On the other hand, the coupling of the OCF game with the ERC game further complicates the analysis.

% which also presents some challenges.
%On the one hand, the ``overlapping'' case increases the analytical complexity in the coalition game.
%On the other hand, the coupling of the OCF game with the ERC game further enhances the difficulty of the analysis.

%%%%%%%%%%%%%%%%%%%%%%%%%%%%%%%%%%%%%%%%%%%%%%%%%%% SYSTEM MODEL %%%%%%%%%%%%%%%%%%%%%%%%%%%%%%%%%%%%%%%%%%%%%%%%%%%%%%%%%
%\input{Section2-Model}

\section{System Model and Problem Formulation}\label{SystemModel}

We consider an MEC-assisted blockchain network as shown in Fig.~\ref{fig:system-model}, which consists of a set $\mathcal{N} \triangleq \{ 1, \ldots, N\}$ of $N$ MUs with limited computing resources serving as miners, and one ECP with substantial computing resources to serve MUs.
The ECP deploys edge computing servers on base stations that are close to the MUs.
Assume that the computing resources of edge computing servers are sufficient for a small area that is not crowded with MUs, such as small residential areas and small business areas in practical scenarios.
Each MU collects the available (i.e., unconfirmed) transactions in the blockchain network and purchases edge computing resources from the ECP to perform the computation-intensive mining task of public blockchain.
Here, edge computing resources can be specified as computing power.
When the mining task is finished, the results will be transmitted back to the MUs.
Suppose that MUs form $M$ coalitions to mine block collaboratively, denoted by the set $\mathcal{C} \triangleq \{ \mathcal{C}_{1},\ldots,\mathcal{C}_{M}\}$, where $\mathcal{C}_{m} \subseteq \mathcal{N} $ is a coalition consisting of a specific subset of MUs.
The key notations  are listed in \textbf{TABLE \ref{table1}}.

\begin{table}[t]
	\begin{center}
		\caption{Summary of Key Notations}
		\label{table1}
		\begin{tabular}{ll}
			\toprule %[2pt]设置线宽
			\textbf{Symbol} & \textbf{Definition} \\
			\midrule
			$\mathcal{N}$  &  The set of MUs  \\
			$N$ &  The number of MUs  \\
			$\mathcal{C}$ &  The set of coalitions formed by MUs  \\
			$M$ &  The number of coalitions  \\
			$\beta_{n,m}$ &  The binary variable of whether MU $n$ joins coalition $\mathcal{C}_{m}$  \\
			$l_{n,m}$ &  The selected nonce length of MU $n$ in coalition $\mathcal{C}_{m}$  \\
			$l_{m}$ &  The selected nonce length of coalition $\mathcal{C}_{m}$  \\
			$r_{n,m}$ &  The expected reward obtained by MU $n$ in coalition $\mathcal{C}_{m}$  \\
			$r_{m}$ &  The expected reward of coalition $\mathcal{C}_{m}$  \\
			$u_{n,m}$ &  The utility of MU $n$ in coalition $\mathcal{C}_{m}$  \\
			$u_{m}$ &  The utility of coalition $\mathcal{C}_{m}$  \\
			$u_{ECP}$ &  The utility of ECP  \\
			$\mathcal{T}_{n}$ &  The set of available transactions collected by MU $n$  \\
			$\mathcal{T}_{m}$ &  The set of available transactions collected by $\mathcal{C}_{m}$  \\
			$\mathcal{N}_{m}$ &  The set of MUs in coalition $\mathcal{C}_{m}$  \\
			$\pi$ &  The adjustable difficulty parameter  \\
			$\phi$ &  The fixed bits length in hash function  \\
			$\lambda$ &  The mean value of Poisson process for solving PoW  \\
			$T'$ &  The average generation time of each block  \\
			$z$ &  The given network latency factor  \\
			$B$ &  The fixed block reward  \\
			$d$ &  The data size of the block header except for nonce  \\
			$\rho_{m}$ &  The transmission power of coalition $\mathcal{C}_{m}$  \\
			$h_{m}$ &  The channel gain of coalition $\mathcal{C}_{m}$  \\
			$N_{0}$ &  The power of noise in the transmission channel \\
			$W$ &  The bandwidth of transmission   \\
			$f_{E}$ &  The CPU computation frequency of ECP  \\
			$\omega$ &  The CPU cycles for each nonce hash computing  \\
			$p$ &  The unit price for nonce hash computing  \\
			$c$ &  The unit cost for nonce hash computing  \\
			$J$ &  The collaboration factor  \\
			
			\bottomrule
		\end{tabular}
	\end{center}
\end{table}

When some MUs form a coalition, they will combine their transaction data together to create a block and offload the PoW mining tasks to the ECP collaboratively.
Note that if an MU $n\in\N$ adopts a solo mining strategy, it can be considered as a coalition with only one member.
It is important to note that the number of coalitions (i.e., $M$) is \emph{not} pre-determined, but changes dynamically depending on decisions of MUs.
Besides, with the multi-coalition collaboration mode, each MU can join multiple coalitions (e.g., $\rm{MU}_{2}$ and $\rm{MU}_{3}$ in Fig.~\ref{fig:system-model}), which implies that different coalitions may overlap with each other, i.e., $\mathcal{C}_{m_1} \bigcap \mathcal{C}_{m_2} $ may not be empty.

To describe a coalition structure, we introduce a binary variable $\beta_{n,m} \in \{0, 1\}$ to denote whether an MU $n$ joins a coalition $\mathcal{C}_{m}$. That is,  $\beta_{n,m}=1$ if MU $n$ joins coalition $\mathcal{C}_{m}$, and $\beta_{n,m}=0$ otherwise.
Obviously, $\sum_{m=1}^{M}\beta_{n,m} \geq 1$ as each MU $n$ can join multiple coalitions.
For notational convenience, we use ${\bm{\upbeta}}_{m} \triangleq (\beta_{1,m},\ldots,\beta_{N,m})$ to denote the decisions of all MUs with respect to coalition $\mathcal{C}_{m}$,
and use ${\bm{\upbeta}}_{n}\triangleq (\beta_{n,1}, \ldots,\beta_{n,M})^\top$ to denote the decisions of MU $n$ with respect to all coalitions.
Moreover, when MU $n$ joins coalition $\mathcal{C}_{m}$, it will contribute its computing power for collaborative block mining.

\subsection{Mining Model}
Each coalition first needs to pack the collected transactions into blocks during the mining process. Let $\mathcal{T}_{n}\triangleq \{t_{n}^{1}, \ldots, t_{n}^{|\mathcal{T}_{n}|}\}$ denote the set of available transactions collected by MU $n$, and $F_n^i$ denotes the transaction fee of a transaction $t_{n}^{i} \in \mathcal{T}_{n}$.
Then, for each coalition $\mathcal{C}_{m}  $, the set of available transactions collected by all MUs in $\mathcal{C}_{m}$ can be written as
$\mathcal{T}_{m}\triangleq  \{t_{m}^{1}, \ldots, t_{m}^{|\mathcal{T}_{m}|}\} = \mathop{\bigcup}_{n \in \mathcal{N}_{m}} \mathcal{T}_{n}$, where $\mathcal{N}_{m} \triangleq \{n\in \mathcal{N}|~\beta_{n,m}=1\}$ is the set of MUs in coalition $\mathcal{C}_{m}$.
Suppose the transaction set $\mathcal{T}_{m}$ is sorted in descending order of transaction fee, i.e., $F_{m}^{1}\geq \cdots \geq F_{m}^{|\mathcal{T}_{m}|}$.
To maximize the expected reward, the top $I$ transactions with the highest transaction fee will be selected to construct a new block.
Thus, the maximum transaction fee that coalition $\mathcal{C}_{m}$ can obtain is:
\begin{equation}\label{}
	\textstyle
	{F(\bm{\upbeta}_{m})} = \sum\limits_{i=1}^{I}F_{m}^{i}.
\end{equation}

When the block is packed, all transactions will form a Merkle tree \cite{30}, where the hash value of the root node will be recorded in the block header, denoted by $X$.
%Let $\mathcal{H}(\cdot)$ represent the hash function, and $X$ denotes the block header data except for nonce.
Such a block header $X$ will be attached with a random number (called \emph{nonce}), and then input into a hash function $\mathcal{H}(\cdot)$ to generate a new hash value.
Each coalition continuously adjusts the nonce to perform repetitive hash computing until the generated hash value satisfies the mining difficulty requirement, such as a specific number of leading zeroes \cite{30}.
The bit length of the nonce, denoted by $\phi$, is determined by the hash function used.
For example, if the hash function is SHA-256 \cite{30}, the bit length of the nonce is $\phi = 32$, and   the search space  is $[0,2^{32} - 1]$.
%Besides, $nonce \in [0,2^{\phi - 1}]$, where $\phi$ is the fixed bits length that represents the search space of the hash function (if the hash function is SHA-256, then $nonce \in [0,2^{32}]$).
A coalition mines a block successfully, only if the generated hash value  satisfies the following condition:
\begin{equation}\label{}
	\textstyle
	\mathcal{H}(X||nonce)\leq V(\pi),
\end{equation}
where $||$ is the merge operation, $V(\pi)=2^{\phi - \pi}$ is the mining difficulty requirement, and $\pi$ is the adjustable difficulty parameter set by blockchain system.
When the generated hash value fails to meet the above requirement, the nonce will be increased by $1$ for repetitive hash computing until a suitable nonce is found.
Since the hash computing performed by each nonce is a memoryless experimental process, each nonce hash computing can be regarded as an i.i.d Bernoulli trial with the following success probability:
\begin{equation}\label{}
	\textstyle
	Pr(\mathcal{H}(X||nonce)\leq V(\pi)) = 2^{-\pi}.
\end{equation}

Due to the different transaction information contained in the block of each coalition, the nonce in the block header is not the same for each block.
After mining a block successfully, it will be propagated to other coalitions for verification.
Due to network latency, not all mined blocks will be verified successfully.
When a block propagates more slowly than other mined blocks, it is at risk of becoming an orphan block and being discarded. Assuming the process of block mining follows a Poisson distribution, the probability of orphan blocks generated by $\mathcal{C}_{m}$ can be computed as follows \cite{31}:

\begin{equation}\label{}
	\textstyle
	{P_{m,orp}} = 1-e^{-\lambda zI},
\end{equation}
where $\lambda=1/T'$ is the mean of Poisson process that model the occurrence of solving the PoW puzzle \cite{31}, $T'$ is the average generation time of each block and $z>0$ is a given network latency factor.
\rev{Typically, we can set $T' = 600$s, which is consistent with the practical Bitcoin network \cite{30}.}
Then, the probability of successful verification can be computed by:
\begin{equation}\label{}
	\textstyle
	{P_{m}} =1-P_{m,orp}=e^{-\lambda zI}.
\end{equation}

Then, the mining reward for each coalition is correlated with the selected nonce length \cite{32}, which can be given by:
\begin{equation}\label{}
	\textstyle
	{r_{m}} = \frac{l_{m}}{l_{\mathcal{N}}}(B+F(\bm{\upbeta}_{m}))e^{-\lambda zI}2^{-\pi},
\end{equation}
where $l_{m} = \sum_{n \in \mathcal{N}_{m}} l_{n,m}$ is the selected nonce length of coalition $\mathcal{C}_{m}$, $l_{n,m} \in \mathbb{Z}$ is the selected nonce length of MU $n$ in coalition $\mathcal{C}_{m}$.
Note that the sequence of nonce selected by members of the same coalition does not overlap. $l_{\mathcal{N}}=\sum_{m = 1}^M l_{m}$ is the total selected nonce length of the whole system, and $B$ is the fixed block reward pre-defined by the blockchain system \cite{31}.
Thus, the reward obtained by each MU $n$ in coalition $\mathcal{C}_{m}$ can be calculated by:
\begin{equation}\label{}
	\begin{aligned}
		\textstyle
		{r_{n,m}} &= \frac{l_{n,m}}{l_{m}}r_{m} = \frac{l_{n,m}}{l_{\mathcal{N}}}(B+F(\bm{\upbeta}_{m}))e^{-\lambda zI}2^{-\pi}.
	\end{aligned}
\end{equation}

\subsection{Offloading Model}
If a MU decides to join a coalition, it sends a request to the head miner of the coalition.
Then, the head miner broadcasts the request to all the members of the coalition.
When all members agree, the head miner will allow the MU to join the coalition.
Due to the limited computing power of mobile devices, the PoW mining tasks of each miner are offloaded to the ECP for nonce hash computing.
In this process, miners in each coalition need to transmit their block header data and announce the range of the selected nonce sequence, e.g., $1 \sim 10^4$.
When the ECP successfully receives this information, it will perform hash computing based on the nonce sequence announced by coalitions.
Let $D_{m}$ denote the data size of the block header, and the range of selected nonce sequences as message data is small that can be neglected.
%Let $d$ denote the data size of the block header except for nonce, and $\eta$ denotes the data size of each nonce.
%Then, the amount of data transmitted to the ECP by each coalition can be calculated as:
%\begin{equation}\label{}
%	\textstyle
%	{D_{m}} = d + \eta l_{m}.
%\end{equation}
Besides, we consider that the signal interference between coalitions during transmission can be neglected.
\rev{Then, the signal-to-noise ratio (SNR) of each coalition $\mathcal{C}_{m} \in \mathcal{C}$ is given by:
\begin{equation}\label{}
	\textstyle
	{\gamma_{m}({\bm{\upbeta}}_{m})} = \frac{\rho_{m}({\bm{\upbeta}}_{m})h_{m}({\bm{\upbeta}}_{m})}{N_{0}},
\end{equation}
where $\rho_{m}({\bm{\upbeta}}_{m})$ is the average transmission power of miners in coalition $\mathcal{C}_{m}$, $h_{m}({\bm{\upbeta}}_{m})$ is the average channel gain of miners in coalition $\mathcal{C}_{m}$, and $N_{0}$ is the noise power in the transmission channel.
Then, the average transmission rate of each coalition can be given by:
\begin{equation}\label{}
	\textstyle
	{R_{m}({\bm{\upbeta}}_{m})} = W\log(1+\gamma_{m}({\bm{\upbeta}}_{m})),
\end{equation}
}where $W$ is the transmission bandwidth. Therefore, the transmission delay for offloading the PoW mining task to the ECP for each coalition is given by:
\begin{equation}\label{}
	\textstyle
	{T_{m}^{tra}} = \frac{D_{m}}{R_{m}({\bm{\upbeta}}_{m})}.
\end{equation}

Assume that the CPU cycles required for each nonce hash computing is $\omega$, and the CPU computation frequency of the ECP is $f_{E}$ (cycles/s).
Here, the edge servers deployed by the ECP have enough CPU cores to support the mining requirements of all coalitions.
Then, the computing delay of the PoW mining task for each coalition can be calculated as follows:
\begin{equation}\label{}
	\textstyle
	{T_{m}^{com}} = \frac{\omega l_{m}}{f_{E}}.
\end{equation}

Since the resultant data returned by the ECP for nonce hash computing is small, the downlink delay can be neglected in this work.
Finally, the overall mining delay includes both transmission delay and computing delay, is given by:
\begin{align}\label{}
	\textstyle
	{T_{m}} &= T_{m}^{tra} + T_{m}^{com} \nonumber \\
	&= \frac{D_{m}}{W\log(1+\gamma_{m}({\bm{\upbeta}}_{m}))} + \frac{\omega l_{m}}{f_{E}}.
\end{align}

\subsection{Problem Formulation}

Considering the cost of mining, the utility of an MU $n$ in coalition $\mathcal{C}_{m}$ can be defined as follows:
\begin{equation}\label{}
	\textstyle
	{u_{n,m}} = \frac{l_{n,m}}{l_{\mathcal{N}}}(B+F(\bm{\upbeta}_{m}))e^{-\lambda zI}2^{-\pi}-l_{n,m}p,
\end{equation}
where $p$ is the unit price for nonce hash computing set by the ECP.
Then, the utility of a coalition $\mathcal{C}_{m}$, denoted by $u_m$, can be defined as the total utility of all MUs in $\mathcal{C}_{m}$, i.e.,
\begin{align}\label{}
	\textstyle
	{u_m} &=\sum\limits_{n \in \mathcal{N}_m} u_{n,m} \nonumber \\
	&=  \sum\limits_{n \in \mathcal{N}_m} {\bigg(\frac{l_{n,m}}{l_{\mathcal{N}}}(B+F(\bm{\upbeta}_{m}))e^{-\lambda zI}2^{-\pi}-l_{n,m}p\bigg)}.
\end{align}

Besides, the system utility is simply defined as the sum utility of all MUs in all coalitions, i.e., $\sum_{n=1}^N \sum_{m=1}^M u_{n,m} $.
Thus, the utility of the ECP can be defined as follows:
\begin{equation}\label{uECP}
	\textstyle
	{u_{ECP}} = \sum\limits_{n=1}^{N} \sum\limits_{m=1}^{M} l_{n,m} p - \sum\limits_{n=1}^{N} \sum\limits_{m=1}^{M} l_{n,m} c,
\end{equation}
where $c$ is the unit cost for nonce hash computing of the ECP.

For convenience, we use $\bm{\upbeta}\triangleq \{\beta_{n,m},\forall n,m\}$ to denote the  decisions   of all MUs with respect to all coalitions, and use $\mathbf{l}\triangleq \{l_{n,m},\forall n,m\}$ to denote the corresponding selected nonce length allocations of different MUs in different coalitions.
Then, the system utility maximization problem can be formulated as follows:
\begin{align}
	&(\mathrm{P1}):~\mathop {\max }\limits_{\bm{\upbeta},\mathbf{l}} ~{\rm{  }}\sum\limits_{n = 1}^N \sum\limits_{m = 1}^M u_{n,m} + u_{ECP}\\
	&~\mathrm{s.t.}~~{\rm{        }}{\sum\limits_{m = 1}^M} {\beta_{n,m}} \leq J, \quad \forall n\in \mathcal{N}, \\
	&~~~~~~~{\rm{        }} 0 \leq {l_{n,m}} \leq {\beta_{n,m}Q}, \quad \forall n \in \mathcal{N}, \mathcal{C}_m\in \mathcal{C}, {l_{n,m}} \in \mathbb{Z}, \\
	&~~~~~~~{\rm{        }} \frac{D_{m}}{W\log(1+\gamma_{m}({\bm{\upbeta}}_{m}))} + \frac{\omega l_{m}}{f_{E}} \leq T', ~~\forall \mathcal{C}_m\in \mathcal{C}, {l_{m}} \in \mathbb{Z}, \\
	&~~~~~~~{\rm{        }} {\beta_{n,m}} \in \{ 0,1\},\quad \forall n \in \mathcal{N}, \mathcal{C}_m\in \mathcal{C},
\end{align}
where $J$ is the collaboration factor indicating the maximum number of coalitions that each MU can join, and $Q$ is a large positive constant. Constraint (17) ensures that each MU $n$ can join no more than the specified number of coalitions.
Constraint (18) ensures that each MU $n$ cannot allocate the selected nonce length to perform the PoW mining task in a coalition $\mathcal{C}_{m}$ that it didn't join, that is, $l_{n,m}=0$ if  $\beta_{n,m}=0$.
Constraint (19) ensures that the overall mining delay of each coalition cannot exceed the average block generation time.
Otherwise, it fails to execute the PoW mining task.

Clearly, the system utility maximization problem (P1) is a centralized optimization problem, which requires  a central controller to make decisions for all MUs in a centralized manner.
However, both MUs and the ECP are independent and selfish, and their objectives are to maximize their own utility, rather than the system utility.
Thus, the centralized solution given in problem (P1) may be inapplicable in practice.
This motivates us to analyze problem (P1) from a game-theoretic perspective.

%Obviously, the optimization problem formulated by (P1) is non-convex and NP-hard. On   one hand,   MUs in the same coalition will share the collected transaction fees, which will be no less than those collected individually. As a result, the reward for each MU in the same coalition is enhanced and MUs have motivation to form coalitions. On the other hand, the cost of computational power rented from ECP is increasing when MU joins multiple coalitions, which reduces the utility of MU to some extent.

%%%%%%%%%%%%%%%%%%%%%%%%%%%%%%%%%%%%%%%%

%%%%%%%%%%%%%%%%%%%%%%%%%%%%%%%%%%%%%%%%%%%%%%%%%%%%%%%%%%%%%%%%%%%%%%%%%%%%%%%%%%%%%%%%%%%%%%%%%%%%%%%%%%%%%%%%%%%%%%%%%%%%%%%%%%%%%%%%%%%%%%%%%%%5

\section{Game Formulation and Analysis}\label{problemFormulation}

In this section, we propose a two-stage Stackelberg game shown in Fig. \ref{fig2} to characterize and analyze the behaviors of MUs and the ECP, where the ECP is the leader and MUs are the followers.
Specifically, in Stage I, the ECP, acting as the leader, determines the prices of computing resources for all MUs.
In Stage II, each MU first selects one or multiple coalitions to join, resulting in an OCF (overlapping coalition formation) game, and subsequently, each coalition decides the amount of edge computing resources to purchase from the ECP, leading to an ERC (edge resource competition) game.
In what follows, we will first formulate the two-stage Stackelberg game, and then analyze the game equilibrium systematically using backward induction.

%In Stage II, we propose a two-layer sequential game to characterize and analyze the collaborative and competitive behaviors of MUs.
%Stage II.A is an OCF game, where each MU acts as a game player and chooses the coalitions to join, forming the coalition structure.
%Stage II.B is an ERC game among the formed coalitions, where each coalition acts as a game player and decides the amount of edge computing resources to invest.
%In Stage I, the ECP decides the resource pricing strategy of the edge computing resources based on the decisions made by MUs in Stage II.
%We will systematically analyze the game equilibrium of the above two-stage Stackelberg game using backward induction.

%%%%%%%%%%%%%%%%%%%%%%%%%%%%%%%%%%%%%%%%%
\begin{figure}[t]
	\centering
	\includegraphics[width=3.1in]{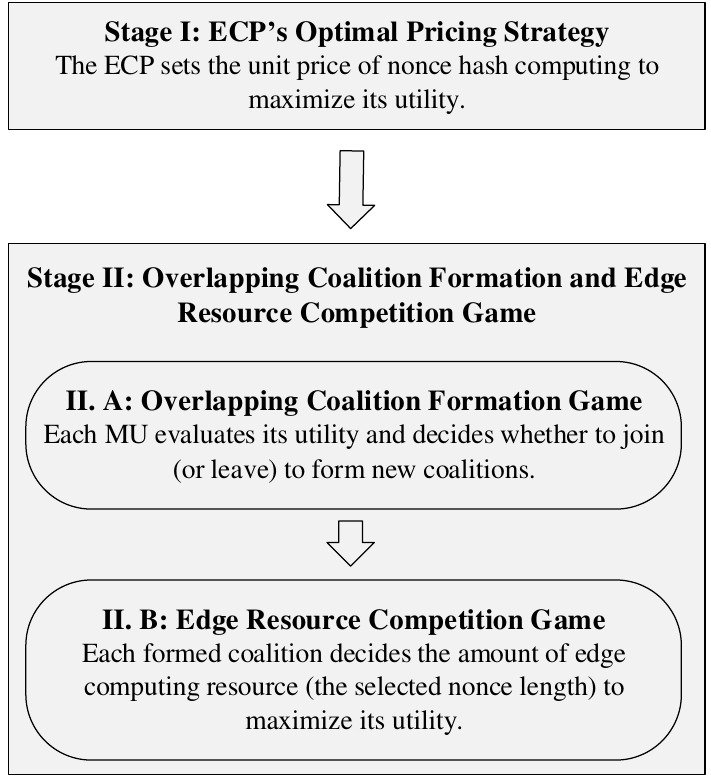}
	\caption{The Proposed Two-Stage Stackelberg Game.}
	\label{fig2}
	\vspace{-3mm}
\end{figure}

\subsection{Two-Stage Stackelberg Game Formulation}

\subsubsection{Stage II: Overlapping Coalition Formation and Edge Resource Competition Game}
The aim of each MU is to decide on coalition joining (or leaving) and edge computing resource investment to maximize its own utility, which can be given as follows:
\begin{align}
	&(\mathrm{P2}):\mathop {\max}\limits_{{\bm{\upbeta}}_{n}, \mathbf{l}_{n}}~~ \sum\limits_{m = 1}^M u_{n,m}
	\\
	&~~~~~~~~~~\mathrm{s.t.}~~{\rm{        }} (17), (18), (19), (20), \nonumber
\end{align}
where $\mathbf{l}_{n} \triangleq \{l_{n,m}, \forall \mathcal{C}_m\in \mathcal{C} \}$ denotes the selected nonce length of MU $n$ in the joined coalitions.

It can be seen that $(\mathrm{P2})$ is intractable due to the coupling of $\bm{\upbeta}_{n}$ and $\mathbf{l}_{n}$.
Therefore, we transform $(\mathrm{P2})$ into a two-layer sequential game to analyze the collaborative and competitive behaviors of MUs, which includes Stage II.A and Stage II.B.
In Stage II.A, each MU decides whether to join (or leave) an existing coalition $\bm{\upbeta}_{n}$ to from new coalitions based on its utility change (see Section \ref{OCF}).
In Stage II.B, each formed coalition from Stage II. A decides the amount of edge computing resource $\mathbf{l}_{m}$ to maximize its utility, where $\mathbf{l}_{m}\triangleq \{l_{n,m}, \forall n\in \mathcal{N}_m  \}$ denotes the selected nonce length allocations of all MUs in coalition $\mathcal{C}_{m}$ (see Section \ref{ERC}).

\subsubsection{Stage I:}
The aim of the ECP is to set the unit price $p$ for nonce hash computing to maximize its utility, which can be given by
\begin{align}\label{uECP_problem}
	&(\mathrm{P3}):~\mathop {\max }\limits_{p} ~{\rm{  }}u_{ECP} \\
	&~~~~~~~~~~~\mathrm{s.t.}~~{\rm{        }} 0 \leq {p} \leq \bar{p},
\end{align}
\rev{where $\bar{p}$ is the maximum price.
Constraint (23) ensures that the set price cannot exceed $\bar{p}$.
Empirically, we set $\bar{p}=500$ in our simulation.}
Problem $(\mathrm{P2})$ and $(\mathrm{P3})$ form a two-stage Stackelberg game, and our goal is to find the Stackelberg equilibrium (SE) for this game.
Consequently, the SE of this game is defined as follows.

\begin{definition}(Stackelberg Equilibrium)
	Let $p^*$ denote the optimal price in Stage I.
	A strategy profile $(\mathbf{l}^{NE}, \bm{\upbeta}^{NE})$ is a Nash Equilibrium (NE) in Stage II, where $\mathbf{l}^{NE} \triangleq \{\mathbf{l}_{1}^{NE}, \mathbf{l}_{2}^{NE},\cdots,\mathbf{l}_{N}^{NE}\}$, and $\bm{\upbeta}^{NE} \triangleq \{\bm{\upbeta}_{1}^{NE}, \bm{\upbeta}_{2}^{NE},\cdots,\bm{\upbeta}_{N}^{NE}\}$
	Then, the point $(p^*, \mathbf{l}^{NE}, \bm{\upbeta}^{NE})$ is the SE, if and only if for any MU $n \in \mathcal{N}$ and the ECP,
	\begin{align}\label{}
		\textstyle
		{u_{n}(\mathbf{l}_{n}^{NE}, \mathbf{l}_{-n}^{NE}, \bm{\upbeta}_{n}^{NE}, \bm{\upbeta}_{-n}^{NE}, p^*)} &\geq {u_{n}(\mathbf{l}_{n}, \mathbf{l}_{-n}^{NE}, \bm{\upbeta}_{n}, \bm{\upbeta}_{-n}^{NE}, p^*)}, \nonumber \\  &\forall \mathbf{l}_{n} \neq \mathbf{l}_{n}^{NE}, \bm{\upbeta}_{n} \neq \bm{\upbeta}_{n}^{NE},
	\end{align}
	and
	\begin{align}\label{}
		\textstyle
		{u_{ECP}(p^*, \mathbf{l}^{NE}(p^*), \bm{\upbeta}^{NE}(p^*))} \geq  &{u_{ECP}(p, \mathbf{l}^{NE}(p), \bm{\upbeta}^{NE}(p))}, \nonumber \\ &\forall p \neq p^*,
	\end{align}
	where $\mathbf{l}_{-n}^{NE} \triangleq \{\mathbf{l}_{1}^{NE}, \cdots, \mathbf{l}_{n-1}^{NE}, \mathbf{l}_{n+1}^{NE}, \cdots, \mathbf{l}_{N}^{NE}\}$, and $\bm{\upbeta}_{-n}^{NE} \triangleq \{\bm{\upbeta}_{1}^{NE}, \cdots, \bm{\upbeta}_{n-1}^{NE}, \bm{\upbeta}_{n+1}^{NE}, \cdots, \bm{\upbeta}_{N}^{NE}\}$.
\end{definition}

In the following, we will analyze the proposed game equilibrium.
First, we analyze Stage II.B in Section \ref{ERC}.
Then, we analyze Stage II.A in Section \ref{OCF}.
Finally, we analyze Stage I in Section \ref{ECP_opt}.

\subsection{The ERC Game Analysis in Stage II.B}\label{ERC}

We first analyze the NE of the ERC game in Stage II.B, given the coalition structure formed by the OCF game in Stage II.A.
%The aim of each coalition is to decide on edge computing resource investment to maximize its respective utility.
Specifically, given the decisions of all MUs with respect to all coalitions $\bm{\upbeta}$, it forms a particular coalition structure $\mathcal{C} \triangleq  \{\mathcal{C}_1, \cdots, \mathcal{C}_M\}$ in II.A.
The objective of each coalition $C_{m} \in \mathcal{C}$ is to maximize the total utility $u_m$ constituted by the MUs it has joined. That is,
\begin{align}
	&(\mathrm{P4}):\mathop {\max}\limits_{\mathbf{l}_{m}}~~ u_{m}
	\\
	&~~~~~~~~~~\mathrm{s.t.}~~{\rm{        }} (19), \nonumber \\
	&~~~~~~~~~~~~~~~~{\rm{        }}~ {l_{n,m}} \geq 0, ~~ \forall n \in \mathcal{N}_m.
\end{align}

It is important to note that the above optimization problem $(\mathrm{P4})$ depends only on the total selected nonce length of coalition $\mathcal{C}_m$, i.e., $l_m = \sum_{n\in \mathcal{N}_m} l_{n,m}$, independent of the detailed selected nonce length allocation among the MUs in $\mathcal{N}_m$.
Thus, we can transform problem $(\mathrm{P4})$ into $(\mathrm{P4'})$, where each coalition $\mathcal{C}_m$ aims to maximize its utility $u_m$ by deciding the total selected nonce length $l_m$. That is,
\begin{align}
	&(\mathrm{P4'}):~\mathop {\max }\limits_{l_{m}} ~{\rm{  }}u_{m} \\
	&~~~~~~~~~~~\mathrm{s.t.}~~{\rm{        }} (19), \nonumber  \\
	&~~~~~~~~~~~~~~~~~{\rm{        }}~ {l_{m}} \geq 0.
\end{align}

Therefore, each coalition $C_{m} \in \mathcal{C}$ competes with each other for the mining reward by deciding the total selected nonce length $l_m$ for MUs in the coalition.
Such a process can be characterized by the following non-cooperative ERC game.~~~~~~~~~~

\begin{definition}
	The ERC game, denoted by $\Omega \triangleq \{\mathcal{C}, \mathbf{l}, \mathbf{U}\}$, is a non-cooperative game defined as follows:
	\begin{itemize}
		\item Game Player: all coalitions in $\mathcal{C} \triangleq \{ C_{1}, \ldots,C_{M}\}$.
		\item Strategy of each coalition $C_{m} \in \mathcal{C}$  is the total selected nonce length $l_{m}$.
		The strategy profile of all coalitions is denoted by $\mathbf{l} \triangleq \{l_{1}, \ldots,l_{M}\}$.
		\item Payoff of each coalition $C_{m} \in \mathcal{C}$ is its utility:
		\begin{equation}\label{um_lm}
			\textstyle
			{u_{m}(l_{m},\boldsymbol l_{-m})} = \frac{l_{m}}{l_{m} + \sum l_{-m}}(B+F(\bm{\upbeta}_{m}))e^{-\lambda zI}2^{-\pi}-l_{m}p,
		\end{equation}
		where $\boldsymbol l_{-m} \triangleq \{l_{1}, \ldots, l_{m-1}, l_{m+1}, \ldots, l_{M}\}$ and $\sum l_{-m} = \sum_{j = 1, j \neq m}^{M} l_{j}$.
		The utility profile of all coalitions is denoted by $\mathbf{U} \triangleq \{u_{1},  \ldots,u_{M}\}$.
	\end{itemize}
\end{definition}

Since $l_{m}$ is a non-negative integer variable, problem $(\mathrm{P4})$ is an integer programming problem, which is intractable.
To achieve efficient processing, we first perform a continuous relaxation of the target variable $l_{m}$. The existence of NE is given in the following theorem.

\begin{theorem}
	There exists an NE in the ERC game $\Omega \triangleq \{\mathcal{C}, \mathbf{l}, \mathbf{U}\}$.
\end{theorem}
\begin{proof}
	See Appendix \ref{Appdendix_A}.
\end{proof}

To obtain the NE of the ERC game, we first derive the best response strategy of each coalition $\mathcal{C}_m$ based on KKT conditions, denoted by $l_m^*$, that is,
\begin{equation}\label{best response}
	\textstyle
	l_m^* = b_m(\boldsymbol l_{-m}) = \left[ \sqrt{\frac{(B+F(\bm{\upbeta}_{m}))e^{-\lambda zI}2^{-\pi} \sum l_{-m}}{p}} - \sum l_{-m} \right]_{0}^{\bar l_m},
\end{equation}
where
\begin{equation}\label{}
	\textstyle
	\bar l_m = \frac{f_E \left( T' W\log(1+\gamma_{m}) - D_{m} \right) }{\omega W\log(1+\gamma_{m})},
\end{equation}
\rev{where $\bar{l}_{m}$ denotes the maximum of the selected nonce length for each coalition.}
By solving the above best response strategies for all coalitions jointly, we can obtain the following properties of NE.
\begin{theorem}
	If $\mathbf{l}^{NE}  $ is the NE of the ERC game, there exists the following conditions hold:
	
	1) For coalition $C_{m}$ with $l_m^* \in [0, \bar l_m]$,
	\begin{equation}\label{NE}
		\textstyle
		\begin{split}
			l_{m}^{NE} = \left\{
			\begin{aligned}
				&\lceil l_{m}^{+} \rceil, ~~~u_{m}(\lceil l_{m}^{+} \rceil) \geq u_{m}(\lfloor l_{m}^{+} \rfloor),\\
				&\lfloor l_{m}^{+} \rfloor, ~~~u_{m}(\lceil l_{m}^{+} \rceil) < u_{m}(\lfloor l_{m}^{+} \rfloor),
			\end{aligned}
			\right.
		\end{split}	
		%		l_{m}^{NE} = \frac{M-1}{p \sum_{j=1}^{M} \frac{1}{\theta_{j}}}\cdot \left[ (1 - \frac{M-1}{\theta_{m} \sum_{j=1}^{M} \frac{1}{\theta_{j}}}) \right]^+, ~~ \forall \mathcal{C}_{m} \in \mathcal{C},
	\end{equation}
	where
	\begin{equation}\label{NE-con}
		l_{m}^{+} =  \frac{M-1}{p \sum_{j=1}^{M} \frac{1}{\theta_{j}}}\cdot  (1 - \frac{M-1}{\theta_{m} \sum_{j=1}^{M} \frac{1}{\theta_{j}}})  , ~~ \forall \mathcal{C}_{m} \in \mathcal{C},
	\end{equation}	
\rev{where $\theta_{m} = (B+F(\bm{\upbeta}_{m}))e^{-\lambda zI}2^{-\pi}$ denotes the expected mining reward of each coalition,} $\lceil z \rceil$ and $\lfloor z \rfloor$ denote that $z$ rounded up and down greedily to the integer solution, respectively.
	
	2) For coalition $C_{m}$ with $l_m^* < 0$, $l_{m}^{NE} = 0$.
	
	3) For coalition $C_{m}$ with $l_m^* > \bar l_m $, $l_{m}^{NE} = \bar l_m$.
\end{theorem}
\begin{proof}
	See Appendix \ref{Appdendix_B}.
\end{proof}

%When the ERC game reaches an NE, the strategies of all coalitions will keep unchanged.
%Then we present the allocation mechanism, where each coalition allocates $u_{m}^{NE}$ in average to its MUs, i.e., $u_{n,m}^{NE} = u_{m}^{NE} / N_{m}$.
\begin{figure*}[ht]
	\centering
	\includegraphics[width=7.1in]{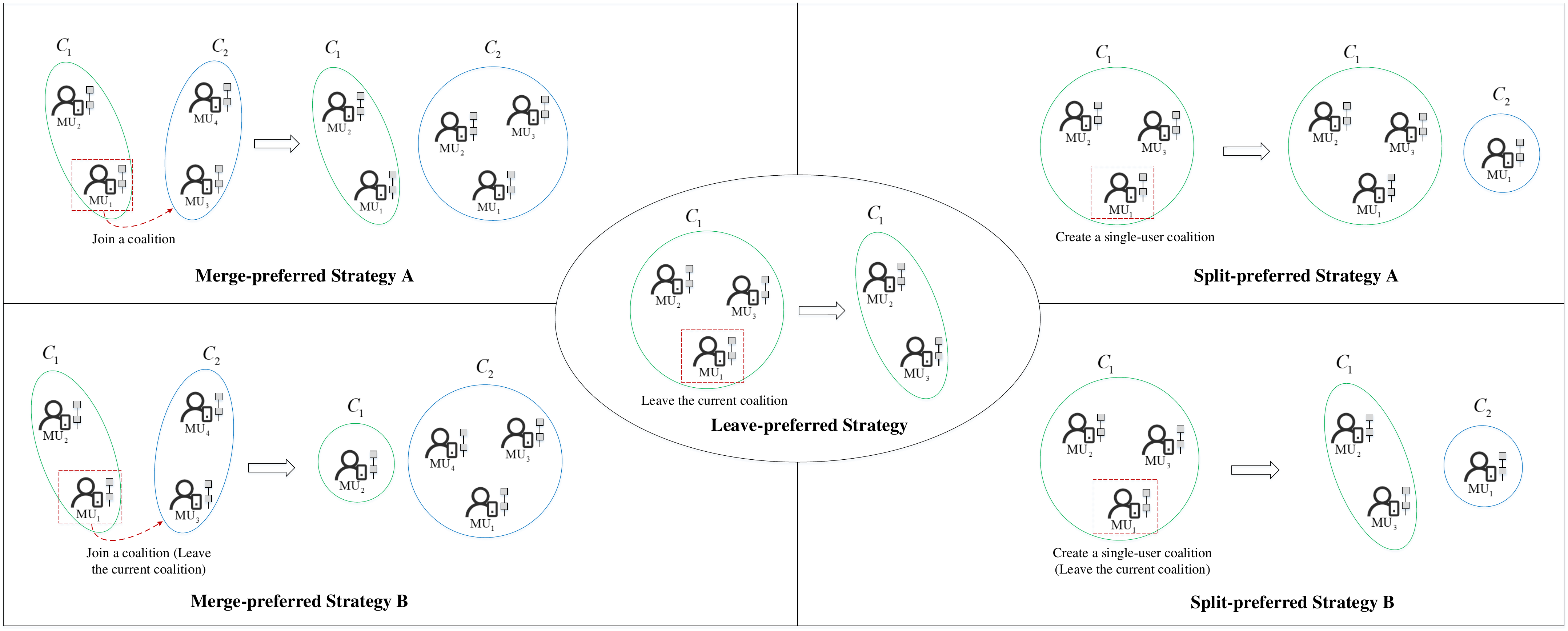}
	\caption{Examples of MU's Atomic Strategies.}
	\label{fig3}
	\vspace{0mm}
	%	Effect of the price of ECP.
\end{figure*}

\begin{lemma}
\rev{The gap bound between $u_{m}(l_{m}^{NE})$ and $u_{m}(l_{m}^{+})$ is as follows.
	\begin{equation}\label{}
		\textstyle
		|u_{m}(l_{m}^{+}) - u_{m}(l_{m}^{NE})| \leq G,
	\end{equation}
	where $G = \max_{l_{m}}|\nabla u_{m}(l_{m})|$.
	}
	
\end{lemma}
\begin{proof}
	See Appendix \ref{Appdendix_C}.
\end{proof}

Note that the NE given in (\ref{NE}) specifies the total selected nonce length of each coalition, but does not determine the detailed allocation of selected nonce length among MUs within the coalition.
As mentioned before, the detailed allocation of selected nonce length (among MUs within a coalition) does not affect the coalition utility or best response strategy.
Therefore, in this work, we adopt a simple \emph{uniform allocation}, where the total selected nonce length of a coalition is evenly divided among all MUs in the coalition, i.e., $ l_{n,m}^{NE} = l_{m}^{NE} / N_{m} $.
Similarly, the utility of a coalition is evenly divided among all MUs in the coalition, i.e., $ u_{n,m}^{NE} =u_{m}^{NE} / N_{m} $.

%%%%%%%%%%%%%%%%%%%%%%%%%%%%%%%%%%%%%%%%

\subsection{The OCF Game Analysis in Stage II.A}\label{OCF}

Now, we analyze the OCF game in Stage II.A.

\subsubsection{OCF Game Definition}

In a classical OCF game, the players collaborate with each other to form coalitions to maximize their individual utilities, where each player is capable of joining multiple coalitions.
%In each coalition, the coalition utility is obtained by aggregating the resources of each player, and the utility obtained by each player in the coalition is fraction of the coalition utility.
%Thus,
Each player decides whether to join a single coalition or multiple coalitions to maximize its total utility.

The OCF game in our model can be defined as follows.
\begin{definition}
	The OCF game is defined as $\mathcal{G} \triangleq (\mathcal{N},v,\mathcal{C})$, where game players are all MUs in $\mathcal{N}$, $v$ is the characteristic function,
	and $\mathcal{C} \triangleq \{ \mathcal{C}_{1},\ldots,\mathcal{C}_{M}\}$ is the coalition structure generated by the OCF game.
\end{definition}

Note that the classical non-overlapping coalition game (where each player can only join one coalition) can be regarded as a special case of the OCF game defined above, by simply letting $ \mathcal{C}_{m}\bigcap \mathcal{C}_{m'}=\emptyset, \forall \mathcal{C}_{m} \neq \mathcal{C}_{m'} \in \mathcal{C} $.
Accordingly, the single-coalition collaboration mode is considered to be a special case of the multi-coalition collaboration mode.

%When the coalition structure is the single-coalition collaboration mode, i.e. $C_{m}\cap C_{m'}=\emptyset,\forall C_{m},C_{m'} \in \mathcal{C}$, the selected nonce length of each player $n$ is allocated to the only coalition it joins. It can be seen that the single-coalition collaboration mode is a special case of the multi-coalition collaboration mode.

The characteristic function $v$ in the OCF game is the utility of game players (i.e., MUs), which is given by:
\begin{equation}\label{equation:characteristic function}
	\textstyle
	\begin{split}
		v_{n,m} = \left\{
		\begin{aligned}
			&u_{n,m}, ~\mbox{if (17), (18) and (19) are satisfied},\\
			&0, \quad \mbox{otherwise}.
		\end{aligned}
		\right.
	\end{split}
\end{equation}

More specifically, if constraints (17), (18) and (19) are satisfied, the characteristic function of MU $n$ denotes the utility received from a coalition $\mathcal{C}_{m}$.
Moreover, if MU $n$ joins multiple coalitions under a coalition structure $\mathcal{C}$, the utility of MU $n$, denoted by $\xi_{n}(\mathcal{C})$, is the total utility received from all coalitions it joins.
\begin{equation}\label{equation:sum utility}
	\textstyle
	\xi_{n}(\mathcal{C}) = \sum\limits_{m = 1}^{M} v_{n,m}, \quad \forall \mathcal{C}_{m} \in \mathcal{C}.
\end{equation}

%Note that the utility in an coalition cannot be apportioned arbitrarily, which is non-transferable utility (NTU) \cite{10}.

Next, we solve for the stable coalition structure, under which no MU has an incentive to deviate from their current strategy.

\subsubsection{MU's Atomic Strategies}

For any coalition structure $\mathcal{C}$, each MU $n$ has five \emph{atomic strategies}:

(i) Deciding whether to join a new coalition $\mathcal{C}_m$ (where $n \notin \mathcal{N}_m$), thereby forming a new coalition $\mathcal{C}_{m'} = \mathcal{C}_m \bigcup \{n\}$ and achieving a new coalition structure $\mathcal{C}'=\mathcal{C} \setminus \{\mathcal{C}_m\} \bigcup \{\mathcal{C}_{m'}\}$;

(ii) Deciding whether to leave an existing coalition $\mathcal{C}_m$ (where $n \in \mathcal{N}_m$) to join a new coalition $\mathcal{C}_{m'}$, thereby forming two new coalition, i.e., $\mathcal{C}_{m''} = \mathcal{C}_m \setminus \{n\}$ and $\mathcal{C}_{m'''} = \mathcal{C}_{m'} \bigcup \{n\}$, and achieving a new coalition structure $\mathcal{C}'=\mathcal{C} \setminus \{\mathcal{C}_{m}, \mathcal{C}_{m'}\} \bigcup \{\mathcal{C}_{m''}, \mathcal{C}_{m'''}\}$;

(iii) Deciding whether to create a new  single-user coalition $\mathcal{C}_{m'} = \{n\}$, thereby achieving a new coalition structure $\mathcal{C}'=\mathcal{C} \bigcup \{\mathcal{C}_{m'}\}$;

(iv) Deciding whether to leave an existing coalition $\mathcal{C}_m$ (where $n \in \mathcal{N}_m$), thereby creating a new coalition $\mathcal{C}_{m'} = \mathcal{C}_m \setminus \{n\}$ and a single-user coalition $\mathcal{C}_{m''} = \{n\}$, and achieving a new coalition structure $\mathcal{C}'=\mathcal{C} \setminus \{\mathcal{C}_m\} \bigcup \{\mathcal{C}_{m'}, \mathcal{C}_{m''}\}$;

(v) Deciding whether to leave an existing coalition $\mathcal{C}_m$ (where $n \in \mathcal{N}_m$), thereby forming a new coalition $\mathcal{C}_{m'} = \mathcal{C}_m \setminus \{n\}$ and achieving a new coalition structure $\mathcal{C}'=\mathcal{C} \setminus \{\mathcal{C}_m\} \bigcup \{\mathcal{C}_{m'}\}$.

It is important to note that if an MU $n$ intends to join a new coalition $\mathcal{C}_m$, the utilities of MUs in $\mathcal{C}_m$ must not be negatively affected, or else MU $n$'s proposal will be rejected.
However, if an MU $n$ desires to leave an existing coalition $\mathcal{C}_m$, it can leave immediately without the consent of other MUs in $\mathcal{C}_m$.
When the coalition structure converges to $\mathcal{C}^{S}$, each MU has no incentive to deviate their strategies to change the coalition structure. In such a case, the coalition structure $\mathcal{C}^{S}$ is \textit{individually stable} \cite{34}.

Formally, we have the following conditions for the above five atomic strategies.

\begin{definition}(Merge-preferred Strategy A) \label{definition5}
	Given any coalition structure $\mathcal{C}$, if an MU $n$ chooses to join a coalition $\mathcal{C}_{m}$ (without leaving the current coalition) to construct a new coalition structure $\mathcal{C}'$, the following conditions must be satisfied:
	(i) the utility of MU $n$ cannot be decreased, i.e. $\xi_{n}(\mathcal{C'}) \geq \xi_{n}(\mathcal{C})$,
	(ii) the utility of MUs in $\mathcal{C}_{m}$ cannot be decreased, i.e., $\xi_{n'}(\mathcal{C'}) \geq \xi_{n'}(\mathcal{C})$, $\forall n'\in \mathcal{N}_{m}$.
\end{definition}

\begin{definition}(Merge-preferred Strategy B) \label{definition6}
	Given any coalition structure $\mathcal{C}$, if an MU $n$ chooses to leave a coalition $\mathcal{C}_{m}$ to join a coalition $\mathcal{C}_{m'}$(where $n \notin \mathcal{N}_{m'}$ and $m \neq m'$), constructing a new coalition structure $\mathcal{C}'$, the following conditions must be satisfied:
	(i) the utility of MU $n$ cannot be decreased, i.e. $\xi_{n}(\mathcal{C'}) \geq \xi_{n}(\mathcal{C})$,
	(ii) the utility of MUs in $\mathcal{C}_{m'}$ cannot be decreased, i.e., $\xi_{n'}(\mathcal{C'}) \geq \xi_{n'}(\mathcal{C})$, $\forall n'\in \mathcal{N}_{m'}$.
\end{definition}

\begin{definition}(Split-preferred Strategy A) \label{definition7}
	Given any coalition structure $\mathcal{C}$, if an MU $n$ chooses to create a single-user coalition $\mathcal{C}_{m'}$ (without leaving the current coalition) to construct a new coalition structure $\mathcal{C}'$, the following conditions must be satisfied:
	the utility of MU $n$ cannot be decreased, i.e. $\xi_{n}(\mathcal{C'}) \geq \xi_{n}(\mathcal{C})$.
	%	Consider a coalition structure $\mathcal{C}=\{ C_{1},\cdots,C_{M}\}$, if MU $n$ chooses to form a singleton coalition $\{n\}$ and the new coalition structure is $\mathcal{C'}=\mathcal{C} \cup \{n\}$, the formation condition is that the utility of MU $n$ is increased, i.e. $\xi_{n}(\mathcal{C'}) > \xi_{n}(\mathcal{C})$.
\end{definition}

\begin{definition}(Split-preferred Strategy B) \label{definition8}
	Given any coalition structure $\mathcal{C}$, if an MU $n$ chooses to leave a coalition $\mathcal{C}_{m}$ (where $n \in \mathcal{N}_m$) to create a single-user coalition $\mathcal{C}_{m'}$, constructing a new coalition structure $\mathcal{C}'$, the following conditions must be satisfied:
	the utility of MU $n$ cannot be decreased, i.e. $\xi_{n}(\mathcal{C'}) \geq \xi_{n}(\mathcal{C})$.
	%	Consider a coalition structure $\mathcal{C}=\{ C_{1},\cdots,C_{M}\}$, if MU $n$ chooses to form a singleton coalition $\{n\}$ and the new coalition structure is $\mathcal{C'}=\mathcal{C} \cup \{n\}$, the formation condition is that the utility of MU $n$ is increased, i.e. $\xi_{n}(\mathcal{C'}) > \xi_{n}(\mathcal{C})$.
\end{definition}

\begin{definition}(Leave-preferred Strategy) \label{definition9}
	Given any coalition structure $\mathcal{C}$, if an MU $n$ chooses to leave a coalition $\mathcal{C}_{m}$ (where $n \in \mathcal{N}_m$), constructing a new coalition structure $\mathcal{C}'$, the following conditions must be satisfied:
	the utility of MU $n$ cannot be decreased, i.e. $\xi_{n}(\mathcal{C'}) \geq \xi_{n}(\mathcal{C})$.
\end{definition}

\begin{algorithm}[t]
	%\textsl{}\setstretch{1.8}
	\renewcommand{\algorithmicrequire}{\textbf{Input:}}
	\renewcommand{\algorithmicensure}{\textbf{Output:}}
	\caption{OCF-based Alternating Iteration Algorithm}
	\label{alg1}
	\textbf{Initialization:} Let $o \triangleq \frac{p}{\bar{p}}$. Set the initial coalition structure as $\mathcal{C} = \mathcal{C}^{0} \triangleq  \{ \{1\}, \ldots,\{N\}\}$, i.e., all MUs work independently, $o = 1$, $o_{pre} = 0$. Set $\tau=0$.
\\
	%\textbf{Overlapping Coalition Formation Stage:} \\
	%	~Perform the following steps based on the current coalition structure $\mathcal{C}$: \\
	\begin{algorithmic}[1]
		\WHILE{$|o - o_{pre}| > \epsilon$}
		\STATE $o_{pre} = o$.
		\FOR{each $p_{k}\in \{(o_{pre}-\Delta^{\tau})\bar{p}, o_{pre}\bar{p}, (o_{pre}+\Delta^{\tau})\bar{p} \}$}
		\REPEAT
		%		\STATE Each MU $n$ tries five profitable strategies in \\ sequence to transform into a new coalition \\ structure $\mathcal{C'}$, and obtain $u_{n,m}^{NE}$ in each coalition.
		\FOR{each MU $n\in \mathcal{N}$}
		\STATE Randomly choose a coalition $\mathcal{C}_m \in \mathcal{C}$, and  select one of the following strategies:
		\STATE \emph{Merge-preferred Strategy A}  in \textbf{Def.~\ref{definition5}}---Join coalition $\mathcal{C}_m$ and construct $\mathcal{C}'$;
%, if (i) conditions in \textbf{Definition \ref{definition5}} are satisfied and  (ii) collaboration factor does not exceed  upper-bound;
		\STATE \emph{Merge-preferred Strategy B} in \textbf{Def.~\ref{definition6}}---Leave coalition $\mathcal{C}_m$ and join   $\mathcal{C}_{m'}$
to construct $\mathcal{C}'$;
%, if (i) conditions in \textbf{Definition \ref{definition6}}  are satisfied and (ii) collaboration factor  does not exceed upper-bound;
		\STATE \emph{Split-preferred Strategy A} in \textbf{Def.~\ref{definition7}}---Create a single-user coalition $\mathcal{C}_{m'}$ and  construct $\mathcal{C}'$;
%, if (i) conditions in \textbf{Definition \ref{definition7}} are satisfied and (ii) collaboration factor does not exceed upper-bound;
		\STATE \emph{Split-preferred Strategy B} in \textbf{Def.~\ref{definition8}}---Leave coalition $\mathcal{C}_m$ and create a single-user coalition $\mathcal{C}_{m'}$ to construct $\mathcal{C}'$;
%, if (i) conditions in \textbf{Definition \ref{definition8}} are satisfied and (ii) collaboration factor does not exceed upper-bound;
		\STATE \emph{Leave-preferred Strategy} in \textbf{Def.~\ref{definition9}}---Leave coalition $\mathcal{C}_m$ to construct $\mathcal{C}'$;
%, if (i) conditions in \textbf{Definition \ref{definition9}} are satisfied;
		\ENDFOR
		%		\STATE MU $n$ joins coalition $C_{m} \in \mathcal{C}$, if (i) \textbf{Definition 4} is satisfied; (ii) MU $n$ has joined less than $M$ \\coalitions.
		%		\STATE When step 2 is not appropriate, MU $n$ forms $\{n\}$, \\if (i) \textbf{Definition 5} is satisfied; (ii) MU $n$ has joined less than $M$ coalitions.
		%		\If{Step 3 Or Step 4 Is Executed \Textbf{And} No Multiple \\ Same Coalitions Exist}
		%		\State $\Mathcal{C} \Gets \Mathcal{C'}$.
		%		\Endif
		\STATE $\mathcal{C} \leftarrow  \mathcal{C'}$.
		\UNTIL (Reach Stable Coalition Structure $\mathcal{C}^{S}$)
		\ENDFOR
		\IF{$u_{ECP}(o_{pre}\bar{p}) \leq u_{ECP}((o_{pre}+\Delta^{\tau})\bar{p})$ and $u_{ECP}((o_{pre}-\Delta^{\tau})\bar{p}) \leq u_{ECP}((o_{pre}+\Delta^{\tau})\bar{p})$}
		\STATE $o = \min(o_{pre}+\Delta^{\tau}, 1)$.
		\ELSIF{$u_{ECP}(o_{pre}\bar{p}) \leq u_{ECP}((o_{pre}-\Delta^{\tau})\bar{p})$ and $u_{ECP}((o_{pre}+\Delta^{\tau})\bar{p}) \leq u_{ECP}((o_{pre}-\Delta^{\tau})\bar{p})$}
		\STATE $o = \max(o_{pre}-\Delta^{\tau}, 0)$.
		\ELSE
		\STATE $o = o_{pre}$.
		\ENDIF
		\STATE \rev{$\Delta^{\tau+1} = 0.99\Delta^{\tau}$.}
		\STATE $\tau = \tau + 1$.
		\ENDWHILE
	\end{algorithmic}
	\textbf{Output:} $p^* = o\bar{p}$, $\mathbf{l}^{NE}(p^*)$ and $\bm{\upbeta}^{NE}(p^*)$.
\end{algorithm}
%\vspace{-3mm}

Fig. \ref{fig3} provides examples to illustrate the MU's atomic strategies.
In these examples, $\rm{MU}_{1}$ chooses to join ${C}_{2}$ from ${C}_{1}$ without leaving ${C}_{1}$ (Definition \ref{definition5}), $\rm{MU}_{1}$ chooses to join ${C}_{2}$ from ${C}_{1}$ and leaves ${C}_{1}$ (Definition \ref{definition6}),
$\rm{MU}_{1}$ chooses to create a single-user coalition ${C}_{2}$ without leaving ${C}_{1}$ (Definition \ref{definition7}),
$\rm{MU}_{1}$ chooses to create a single-user coalition ${C}_{2}$ and leaves ${C}_{1}$ (Definition \ref{definition8}),
and $\rm{MU}_{1}$ chooses to leave the current coalition ${C}_{1}$ (Definition \ref{definition9}).
The following  \textbf{Theorem \ref{Convergence}} shows the convergence of these atomic strategies.
\begin{theorem}(Convergence) \label{Convergence}
	Through the MUs' atomic strategies, arbitrary initial coalition structures will converge to a stable coalition structure $\mathcal{C^{\mathit{S}}}$.
\end{theorem}

\begin{proof}
	See Appendix \ref{Appdendix_D}.
\end{proof}

\begin{table*}[t]
	\begin{center}
		\caption{Simulation Parameters}
		\label{Parameters}
		\begin{tabular}{llll}
			\toprule %[2pt]设置线宽
			\textbf{Parameter} & \textbf{Value} & \textbf{Parameter} & \textbf{Value}\\
			\midrule
			The given network latency factor, $z$ &  $5 \times 10^{-3}$ & The average generation time of each block, $T'$ &  $600$ s\\
			The data size of the block header except for nonce, $D_{m}$ &  $608$ bit & The unit cost for nonce hash computing, $c$ &  $0.8$ \\
			The adjustable difficulty parameter, $\pi$ &  0.5 & The transmission power of coalition $\mathcal{C}_{m}$, $\rho_{m}$ &  $0.1$ W \\
			The channel gain of coalition $\mathcal{C}_{m}$, $h_{m}$ &  $1 \times 10^{-8}$  & The power of noise in the transmission channel, $N_{0}$ & $-100$ dBm \\
			The CPU computation frequency of ECP, $f_{E}$ &  $1$ GHz   & The bandwidth of transmission, $W$ &  $20$ MHz  \\
			The CPU cycles for each nonce hash computing, $\omega$ &  $1 \times 10^{3}$ Mega cycles  &   &     \\
			\bottomrule
		\end{tabular}
	\end{center}
	\vspace{-5mm}
\end{table*}

\subsection{The ECP's Optimal Pricing Strategy in Stage I}\label{ECP_opt}
Now we analyze the ECP's optimal strategy in Stage I.
Based on the $\mathbf{l}^{NE}$ of all formed coalitions in Stage II, the ECP maximize its utility by optimizing its unit price for once hash computing.
We apply ${l}^{NE}_{m}(p)$ into Eq. (\ref{uECP_problem}) to obtain the following optimization problem $(\mathrm{P5})$, which is as follows:
\begin{align}\label{uECP_expression}
	&(\mathrm{P5}):~\mathop {\max }\limits_{p} ~{\rm{  }} u_{ECP} = \sum\limits_{m=1}^{M} l^{NE}_{m}(p) \cdot ( p- c) \\
	&~~~~~~~~~~~\mathrm{s.t.}~~{\rm{        }} (23). \nonumber
\end{align}

Note that the objective in $(\mathrm{P5})$ cannot have an explicit expression since ${l}^{NE}_{m}(p)$ depends on $p$ and $\bm{\upbeta}_{m}$.
In this case, we adopt the sub-gradient search \cite{37} for $p$ in the feasible domain $[0, \bar{p}]$.

Based on the above five atomic strategies of MUs and the sub-gradient search for $p$, we propose an OCF-based alternating iteration algorithm to solve the two-stage Stackelberg game, as described in \textbf{Algorithm \ref{alg1}}, where $\Delta$ is the step size, $\tau$ is the number of iterations, and $\epsilon$ is the convergence threshold.
Specifically, in each round (iteration) of Algorithm \ref{alg1} with a given $p_k$, each MU randomly chooses a coalition and tests the above five atomic strategies.

Since each MU $n$ is able to implement the above five atomic strategies without relying on any centralized entity, the proposed OCF-based alternating iteration algorithm can be implemented in a distributed approach.
Under the current coalition structure $\mathcal{C}$, each MU $n$ needs to evaluate its respective utility $\xi_{n}(\mathcal{C})$ through message propagation with other members of the coalition it joins.
In addition, given a current coalition structure $\mathcal{C}$ and the five atomic strategies, the computational complexity for each MU $n$ to find the next coalition does not exceed $\mathcal{O}(|\mathcal{C}|)$ \cite{25}.
\rev{Moreover, we analyze the overall time complexity of Algorithm 1.
	Assume that the coalition structure requires $\Gamma_{1}$ iterations to reach the stable state.
	Each MU will try five atomic strategies and select one of them.
	Therefore, the complexity of the inner loop is $\mathcal{O}(5N\cdot |\mathcal{C}|\cdot\Gamma_{1})$.
	Next, we assume that the outer loop requires $\Gamma_{2}$ rounds of iterations to reach the convergence condition.
	Each iteration will try three pricing strategies.
	Then, the overall time complexity of Algorithm 1 is $\mathcal{O}(15N\cdot|\mathcal{C}|\cdot\Gamma_{1}\Gamma_{2})$.}

\begin{figure*}[htb]
	\centering
	\hspace{-1mm}
	\subfigure[]{\label{fig:price_ecputi}\includegraphics[width=3in]{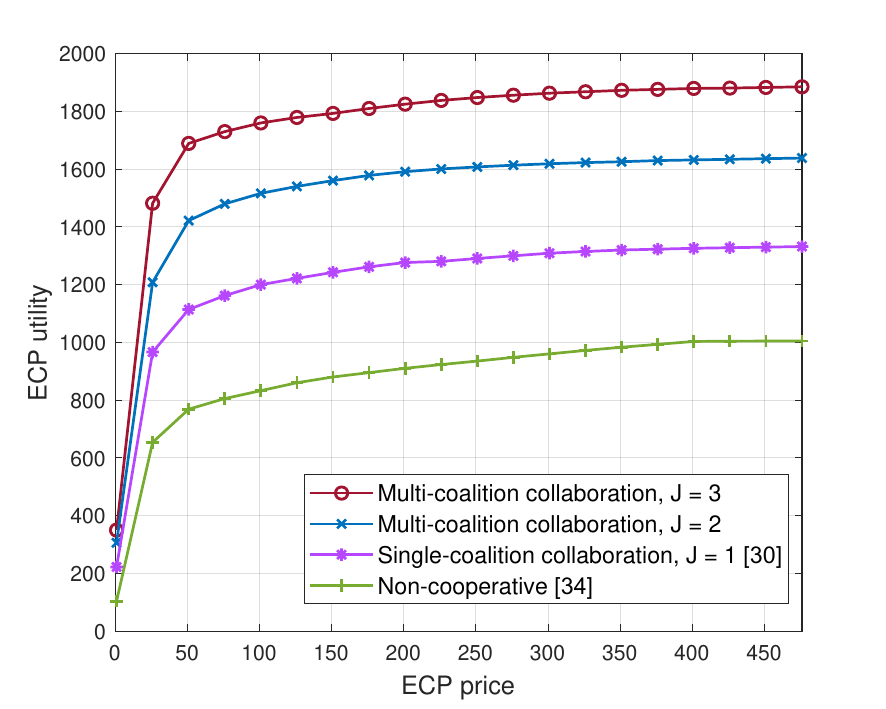}}
	~~
	\subfigure[]{\label{fig:price_sysuti}\includegraphics[width=3in]{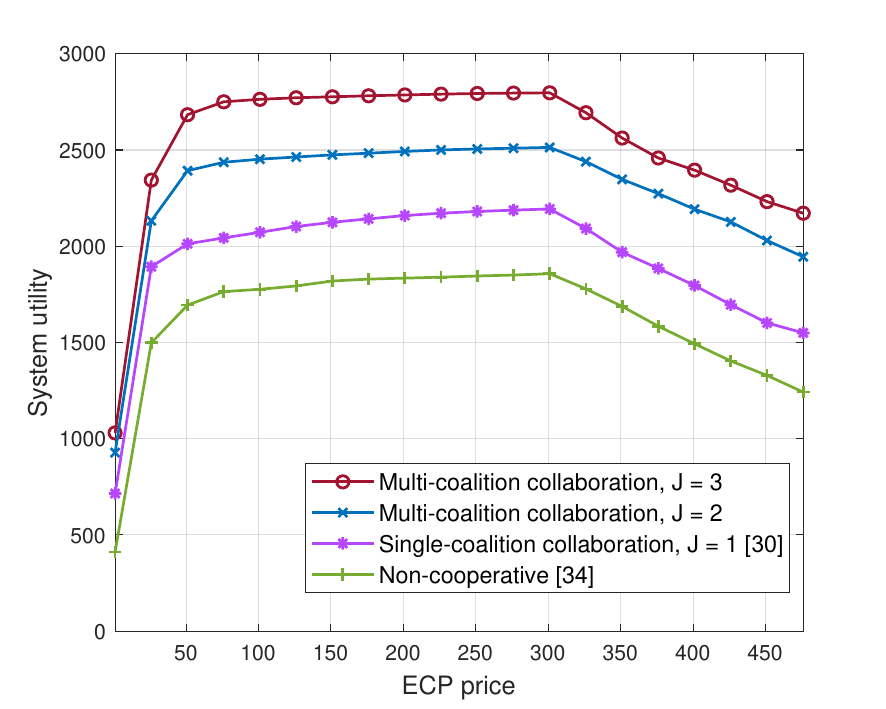}}
	~~
	\subfigure[]{\label{fig:price_NonceLength}\includegraphics[width=3in]{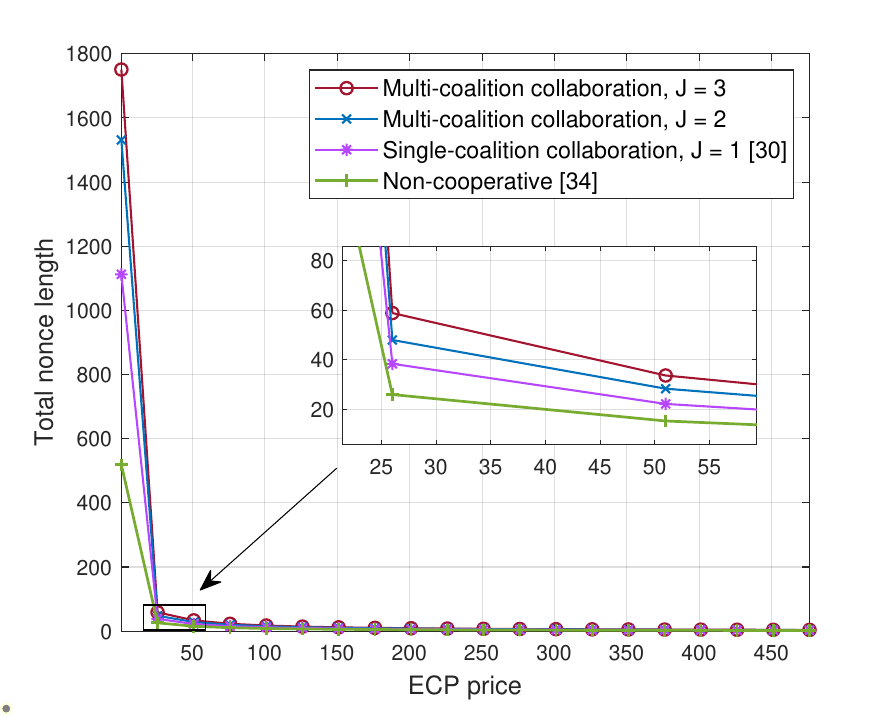}}
	~~
	\subfigure[]{\label{fig:price_grouping}\includegraphics[width=3in]{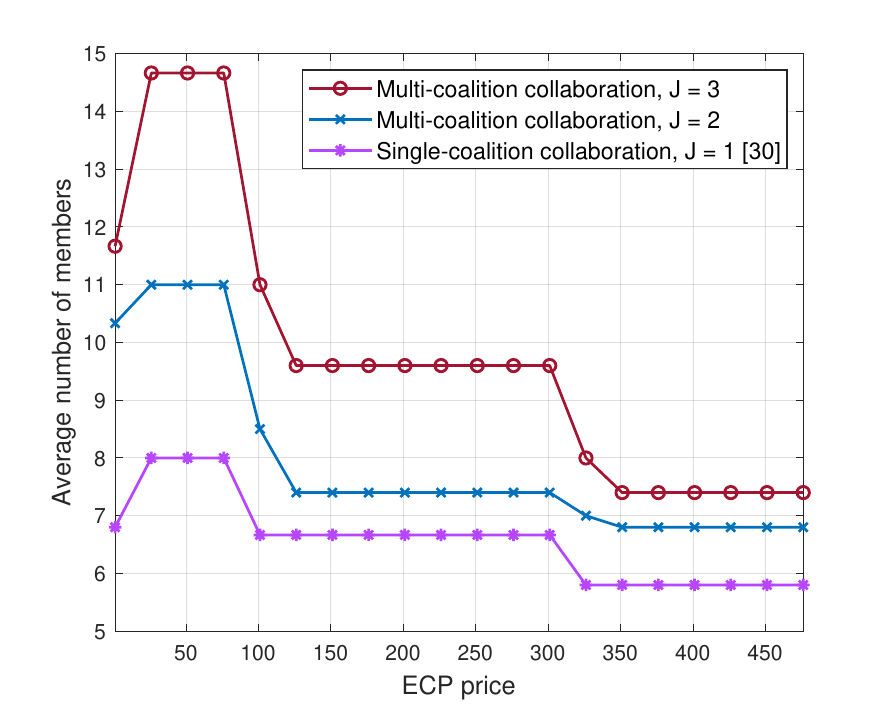}}
	\caption{Performance of different prices.  (a) ECP utility; (b) System utility; (c) Total nonce length of MUs; (d) Average number of members.}
	\label{fig:Effect of the price of ECP}
	%	Effect of the price of ECP.
\end{figure*}

\begin{figure*}[htb]
	\centering
	\hspace{-1mm}
	\subfigure[]{\label{fig:MU_ecputi}\includegraphics[width=3in]{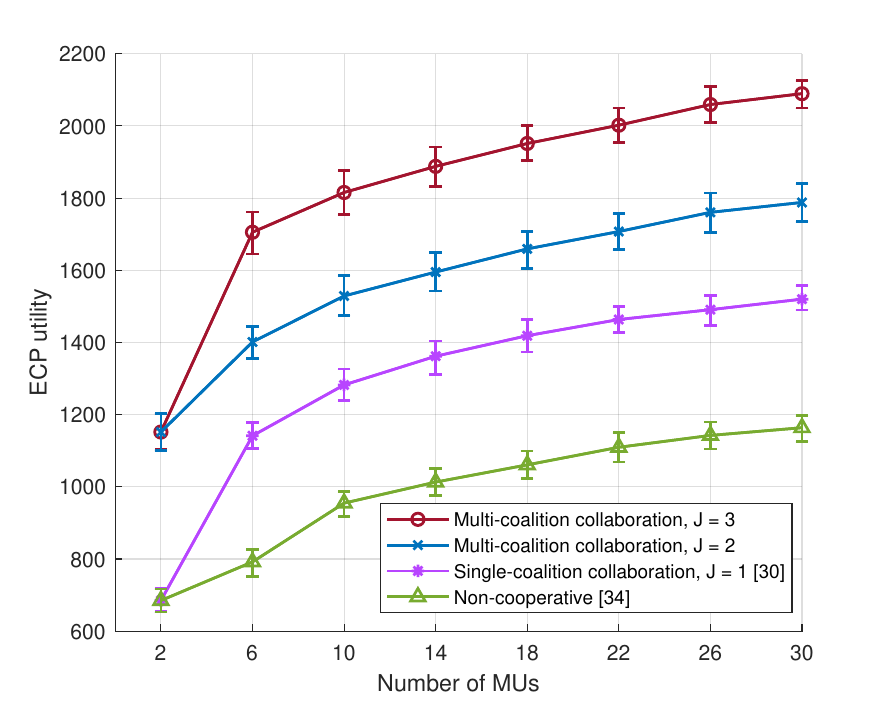}}
	~~
	\subfigure[]{\label{fig:MU_sysuti}\includegraphics[width=3in]{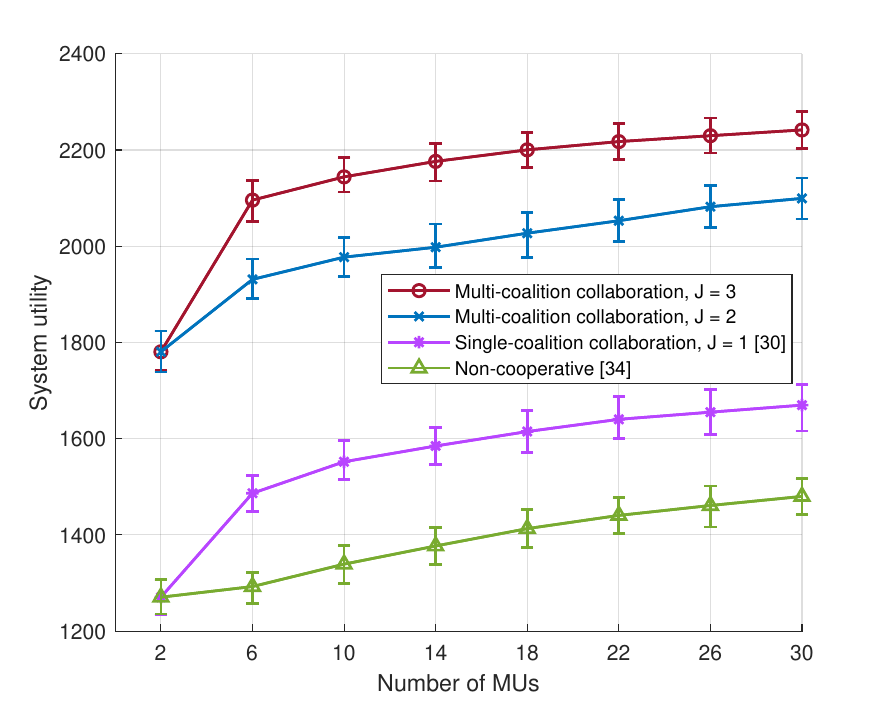}}
	~~
	\subfigure[]{\label{fig:MU_NonceLength}\includegraphics[width=3in]{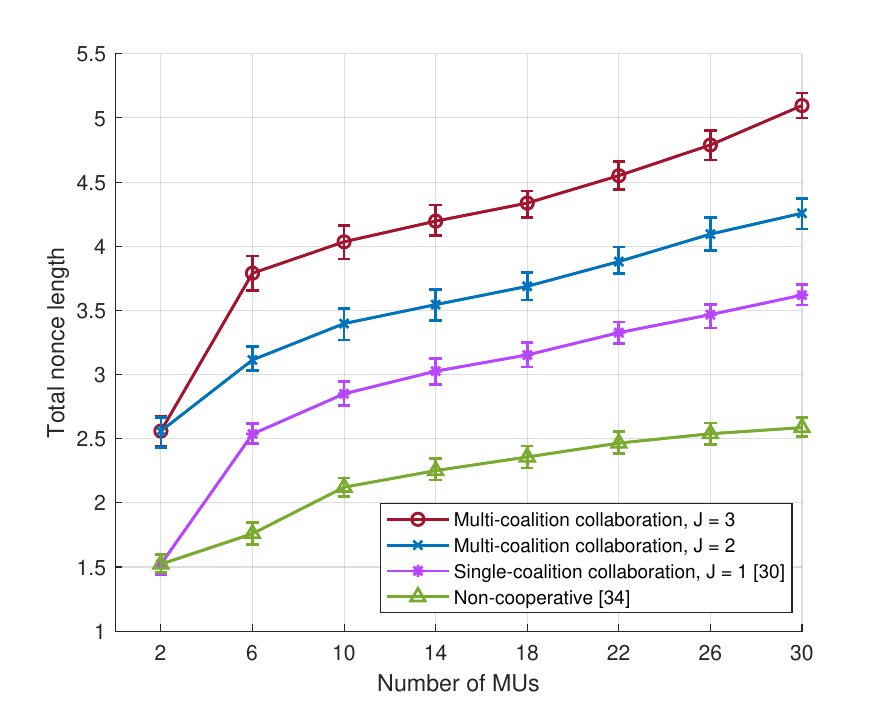}}
	~~
	\subfigure[]{\label{fig:MU_grouping}\includegraphics[width=3in]{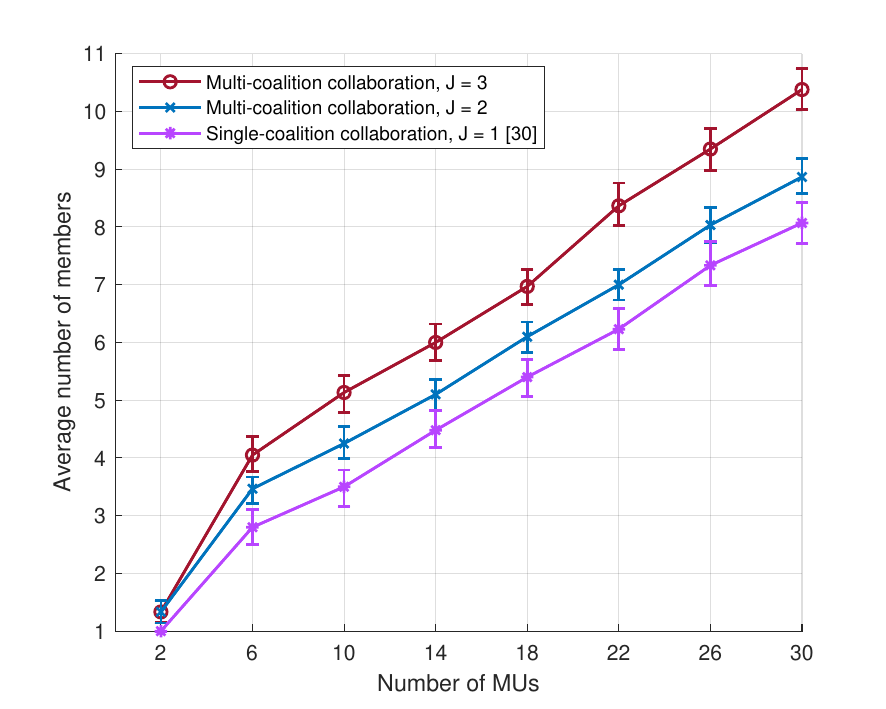}}
	\caption{Impact of the number of MUs. (a) ECP utility;  (b) System utility; (c) Total nonce length of MUs; (d) Average number of members.}
	\vspace{-3mm}
	\label{fig:Effect of the number of MUs}
\end{figure*}

%\begin{table*}[htb]
%	\begin{center}
%		\caption{Simulation Parameters}
%		\label{Parameters}
%		\begin{tabular}{ll}
%			\toprule %[2pt]设置线宽
%			\textbf{Parameter} & \textbf{Value} \\
%			\midrule
%			The given network latency factor, $z$ &  $5 * 10^{-3}$  \\
%			The average generation time of each block, $T'$ &  $600$ s  \\
%			The data size of the block header except for nonce, $D_{m}$ &  $608$ bit  \\
%			The unit cost for nonce hash computing, $c$ &  $0.8$  \\
%			The adjustable difficulty parameter, $\pi$ &  0.5  \\
%			The transmission power of coalition $\mathcal{C}_{m}$, $\rho_{m}$ &  $0.1$ W  \\
%			The channel gain of coalition $\mathcal{C}_{m}$, $h_{m}$ &  $1*10^{-8}$  \\
%			The power of noise in the transmission channel, $N_{0}$ & $-100$ dBm  \\
%			The bandwidth of transmission, $W$ &  $20$ MHz  \\
%			The CPU cycles required for each nonce hash computing, $\omega$ &  $1 * 10^{3}$ Mega cycles  \\
%			The CPU computation frequency of ECP, $f_{E}$ &  $100$ GHz  \\

%			\bottomrule
%		\end{tabular}
%	\end{center}
%\end{table*}

%\input{Section5-DRL}

%%%%%%%%%%%%%%%%%%%%%%%%%%%%%%%%%%%%%%%%%%%%%%%%%%%%%%%%%%%%%%%%%%%%%%%%%%%%%%%%%%%%%%

\vspace{-1mm}
\section{Performance Evaluation}\label{simulation}

In this section, we provide extensive results to evaluate the performance of the proposed OCF-based alternating iteration algorithm.
We first present the simulation settings of the results.
Then, we show the performance of different prices in Stage II as the intermediate result.
Next, we show the impact of the number of MUs on the performance of the proposed algorithm.
Finally, we present the impact of different parameters on the performance of the proposed algorithm.

\subsection{Simulation Setting}
The simulation parameter settings is partly referred to \cite{30}, as shown in \textbf{TABLE \ref{Parameters}}.
In the simulation parameters, some of them are related to the real-world blockchain network, such as $D_{m}$, $z$, and $T'$, etc.
We assume that the blockchain system generates $1000$ transactions for MUs to collect, where each transaction fee follows uniform a distribution of $[0,100]$ units \cite{32}.
To evaluate the proposed algorithm, we set the collaboration factor $J=1,2,3$, where $J=1$ denotes the single-coalition collaboration mode \cite{25}.
Besides, the benchmark for comparison is the non-cooperative mode \cite{30}, in which each MU refuses to form coalitions.
All simulation results are averaged over $500$ times.
To evaluate the atomic strategies of MUs, we introduce a new metric, i.e., the average number of members in each coalition, which is given by:
\begin{equation}\label{equation:average member}
	\textstyle
	N_{avg} = \frac{\sum_{m = 1}^{M} |\mathcal{N}_{m}|}{M},  \forall \mathcal{C}_{m} \in \mathcal{C}.
\end{equation}

\subsection{Performance of Different Prices in Stage II}

Fig.~\ref{fig:Effect of the price of ECP} shows the performance of different prices $p$ in Stage II as the intermediate result, where the number of MUs is $N=20$, the fixed block reward is $B=1000$, and each block contains $I=10$ transactions.
The horizontal axis of the four sub-figures denotes the price of ECP $p$.

Fig.~\ref{fig:price_ecputi} and Fig.~\ref{fig:price_NonceLength} show the ECP's utility $u_{ECP}$, and the total nonce length of MUs $\sum_{m=1}^{M} l_{m}$  respectively under different $p$.
It can be seen that $u_{ECP}$ increases and $\sum_{m=1}^{M} l_{m}$ decreases as $p$ grows.
Specially, $\sum_{m=1}^{M} l_{m}$ first decreases rapidly and then tends to stabilize.
This is because the increase in $p$ leads to higher mining costs for each MU, which reduces the selected nonce length that is offloaded to the ECP.
When $p \leq 50$, the decrease rate of $\sum_{m=1}^{M} l_{m}$ is lower than the increase rate of $p$, resulting in $u_{ECP}$ growing rapidly;
When $p > 50$, the decrease rate of $\sum_{m=1}^{M} l_{m}$ gradually equalizes with the increase rate of $p$, causing  $u_{ECP}$ to stabilize (based on Eq. (\ref{uECP_expression})).
Compared with the non-cooperative mode, the single-coalition collaboration mode increases the ECP's utility by $32.67\%$.
Meanwhile, the multi-coalition collaboration mode with $J = 3$ can further enhances the ECP's utility by $41.52\%$.

Fig.~\ref{fig:price_sysuti} shows the blockchain system's utility under different $p$.
When $p \leq 300$, we can see that the system's utility increases as $p$ grows.
This implies that the ECP's utility plays a dominant role in the system's utility, and the growth of $p$ greatly enhances the ECP's utility, which in turn facilitates the system's utility;
When $p > 300$, we can see that the system's utility decreases rapidly as $p$ grows.
This implies that the MUs' total utility plays a dominant role in the system's utility, and the growth of $p$ leads to the reluctance of almost all MUs to participate in block mining.
Compared with the non-cooperative mode, the single-coalition collaboration mode increases the system's utility by $18.09\%$.
Meanwhile, the multi-coalition collaboration mode with $J = 3$ can further enhances the system's utility by $33.35\%$.

Fig.~\ref{fig:price_grouping} shows the average number of coalition members $N_{avg}$ under different $p$.
Note that there are no formed coalitions in the non-cooperative mode, and hence, there is no curve for the non-cooperative mode in the figure.
It can be seen that the $N_{avg}$ increases briefly and then declines in a segmental manner as $p$ grows, and the multi-coalition collaboration mode with $J = 3$ outperforms the other modes.
When $p$ is low, e.g., $p \leq 25$, MUs prefer to cooperate in order to increase their revenue as $p$ grows.
This is due to the fact that MUs can keep their own revenue high by not cooperating when $p$ is low, so they have no willingness to cooperate.
When $p$ is high, e.g., $p > 25$, MUs gradually withdraw from the joined coalitions due to the increase in mining costs and significantly reduce the purchase of computing resources.

%When $p > 300$, MUs further withdraw from the joined coalitions, and the huge mining costs make almost all MUs reluctant to buy computing resources from ECP.

\subsection{Impact of the Number of MUs}

Fig.~\ref{fig:Effect of the number of MUs} shows the impact of the number of MUs on the performance of the proposed OCF-based algorithm.
The horizontal axis of the four sub-figures denotes the number of MUs $N$.
Here, the fixed block reward is $B=1000$, and each block contains $I=10$ transactions.
All curves are with confidence intervals.

Fig.~\ref{fig:MU_ecputi} and Fig.~\ref{fig:MU_NonceLength} show the ECP's utility $u_{ECP}$ and the total nonce length of MUs $\sum_{m=1}^{M} l_{m}$ respectively under different $N$.
It can be seen that $u_{ECP}$ and $\sum_{m=1}^{M} l_{m}$ both increase as $N$ grows.
Specifically, the growth of $N$ makes MUs compete for computing resources more intensely, which leads to a higher demand for computing resources.
Intuitively, the ECP's utility also increases.
Compared with the non-cooperative mode, the single-coalition collaboration mode increases the ECP's utility by $11.85\% \sim 14.51\%$.
Meanwhile, the multi-coalition collaboration mode with $J = 3$ can further enhances the ECP's utility by $12.74\% \sim 16.96\%$.

Fig.~\ref{fig:MU_sysuti} shows the blockchain system's utility under different $N$.
We can see that the blockchain system's utility increases as $N$ grows.
This is because the ECP's utility plays a dominant role in the system's utility, and the growth of $N$ greatly enhances the ECP's utility, which also facilitates the system's utility.
Compared with the non-cooperative mode, the single-coalition collaboration mode increases the ECP's utility by $10.42\% \sim 12.48\%$.
Meanwhile, the multi-coalition collaboration mode with $J = 3$ can further enhance the ECP's utility by $12.64\% \sim 17.63\%$.

Fig.~\ref{fig:MU_grouping} shows the average number of coalition members $N_{avg}$ under different $N$.
We can see that $N_{avg}$ increases as $N$ grows.
This is because the increase in $N$ makes competition for computing resources more intense, leading to a greater tendency for MUs to cooperate in forming coalitions.
Since each MU can choose to join more coalitions in the multi-coalition collaboration mode with $J = 3$, $N_{avg}$ is significantly larger than the other modes.

\rev{Fig.~\ref{fig:converage_p} shows the convergence of the subgradient search for $p$.
Here, the number of MUs is $N=20$, the fixed block reward is $B=1000$, and each block contains $I = 10$ transactions.
It can be seen that $p$ converges to an optimum after a finite number of iterations.
As the number of iterations increases, the amplitude of the oscillations of all curves gradually decreases, which is due to the reduction of the step size $\Delta$.}
Fig.~\ref{fig:MU-optprice} shows the optimal price $p^{*}$ under different $N$.
When $N$ is less than $18$, the optimal price remains unchanged.
When $N$ is greater than $18$, the optimal price in the collaborative mode begins to decrease, while in the non-cooperative mode it remains unchanged.
This indicates that when $N$ is large, the collaborative mode significantly  reduces the demand for resources. Consequently, the ECP appropriately adjusts the price downward to prevent its utility from declining.
Furthermore, it also shows that as the collaboration factor $J$ increases, the demand for resources decreases, and therefore the ECP appropriately reduces $p$ to maintain the demand of MUs.

\subsection{Impact of Other Parameters}
Fig.~\ref{fig:Effect of tx number} shows the blockchain system's utility under different number of transactions $I$ packed into the block, where the number of MUs is $N=20$ and the fixed block reward is $B=1000$.
We can see that the blockchain system's utility increases as $I$ grows.
This is because each MU can collect more transactions with a larger $I$, and the multi-coalition collaboration mode can aggregate more profitable transactions, thereby increasing coalition utility.
Compared with the non-cooperative mode, the single-coalition collaboration mode increases the ECP's utility by $35.41\% \sim 86.67\%$.
Meanwhile, the multi-coalition collaboration mode with $J = 3$ can further enhance the ECP's utility by $17.39\% \sim 27.92\%$.

Fig.~\ref{fig:Effect of block reward} shows the blockchain system's utility under different fixed block reward $B$, where the number of users is $N=20$ and each block contains $I=10$ transactions.
We can see that the blockchain system's utility increases as $B$ grows.
Intuitively, this is due to the fact that an increase in $B$ leads to a rise in mining rewards, which in turn promotes the system's utility.
Compared with the non-cooperative mode, the single-coalition collaboration mode increases the ECP's utility by $18.63\% \sim 23.48\%$.
Meanwhile, the multi-coalition collaboration mode with $J = 3$ can further enhance the ECP's utility by $14.22\% \sim 33.31\%$.

\begin{figure}[t]
	\centering
	\includegraphics[width=3.1in]{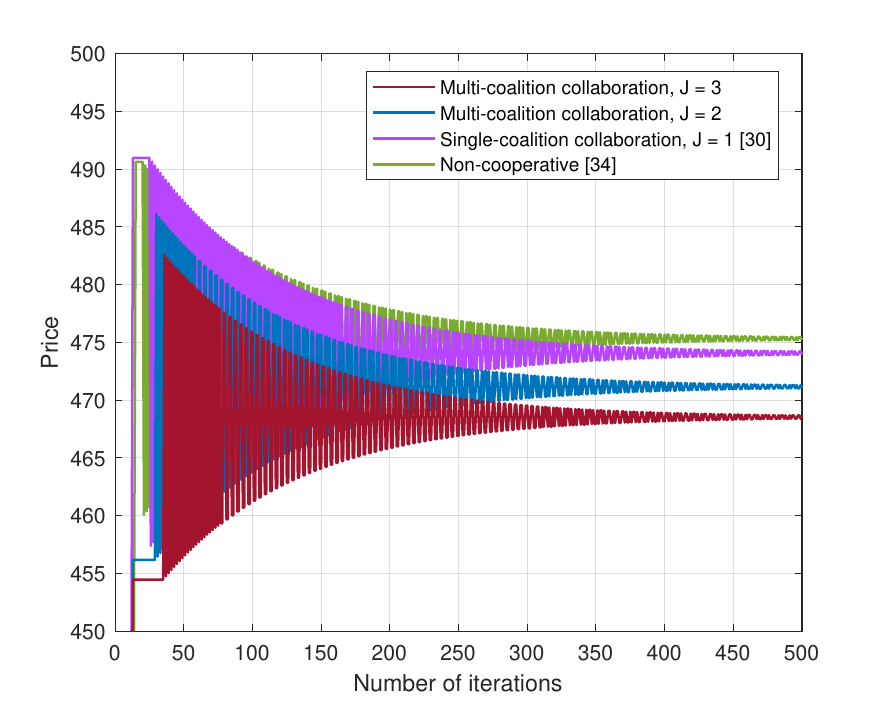}
	\vspace{-2mm}
	\caption{Price versus Number of iterations.}
	\label{fig:converage_p}
	\vspace{-4mm}
\end{figure}

\begin{figure}[t]
	\centering
	\includegraphics[width=3.1in]{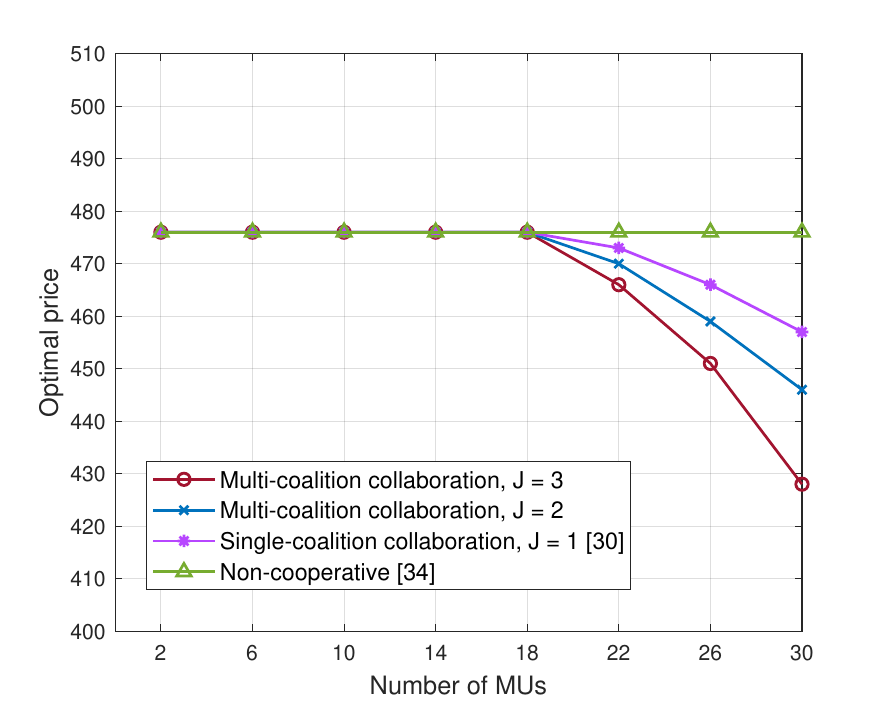}
	\vspace{-2mm}
	\caption{Optimal price versus Number of MUs.}
	\label{fig:MU-optprice}
	\vspace{-4mm}
\end{figure}

\begin{figure}[t]
	\centering
	\includegraphics[width=3.1in]{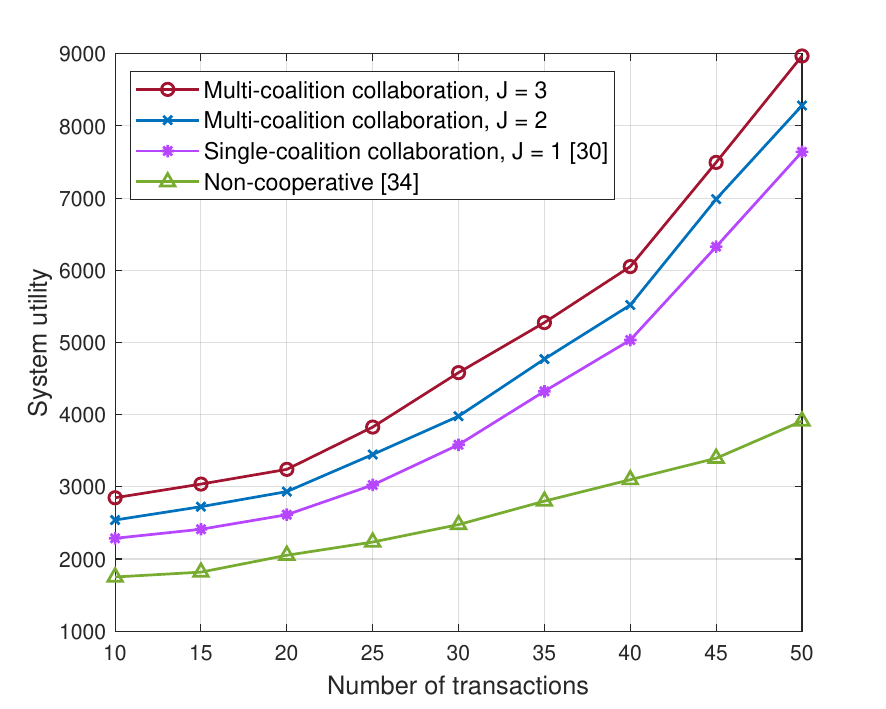}
	\vspace{-2mm}
	\caption{System utility versus Number of transactions.}
	\label{fig:Effect of tx number}
	\vspace{-4mm}
\end{figure}

\begin{figure}[t]
	\centering
	\includegraphics[width=3.1in]{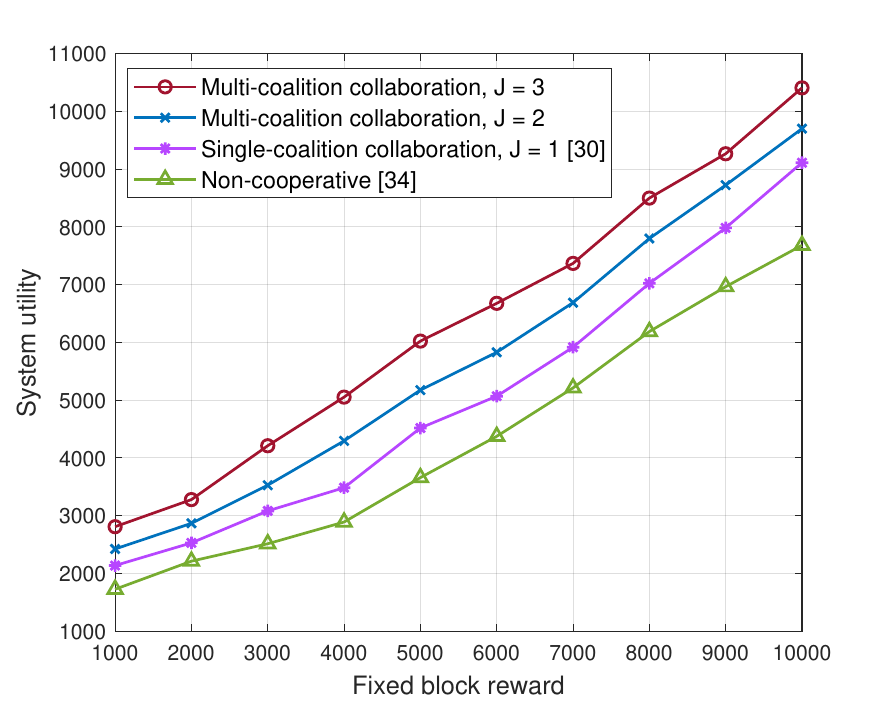}
	\vspace{-2mm}
	\caption{System utility versus Fixed block reward.}
	\label{fig:Effect of block reward}
	\vspace{-4mm}
\end{figure}

%%%%%%%%%%%%%%%%%%%%%%%%%%%%%%%%%%%%%%%%%%%%%%%%%%%%%%%%%%%%%%%%%%%%%%%%%%%%%%%%%%%%%%%%%%%%%%%%%%
\section{Conclusion}\label{conclusion}

In this work, we have studied an MEC-assisted collaborative blockchain network with a multi-coalition collaboration mode.
To analyze the behavior of miners and the ECP in such a scenario, we have proposed a two-stage Stackelberg game, consisting of a resource pricing problem (for the ECP) in Stage I, and a coalition formation game (for the miners) and a resource competition game (for the coalitions) in Stage II.
We have derived the closed-form Nash equilibrium for the ERC game, and proposed an OCF-based alternating algorithm that converges to a stable coalition structure for the OCF game, as well as a near-optimal pricing strategy for the resource pricing problem.
Simulation results have shown that the proposed multi-coalition collaboration mode can significantly enhance block mining efficiency.
Moreover, we find that when the price is low, MUs' willingness to cooperate gradually increases as the price rises, and when the price is high, MUs' willingness to cooperate decreases gradually as the price rises.

\rev{Furthermore, this work focuses on mining efficiency optimization without considering security and power consumption issues.
Specifically. edge servers may be vulnerable to potential attacks such as Distributed Denial of Service (DDoS) attacks, eavesdropping, and side-channel attacks, etc. How to prevent these potential attacks while optimizing performance is a future research direction.
Additionally, reducing the system power of nonce hash computing while improving the mining efficiency is also a worthwhile research direction.}

\appendices
\section{Proof for Theorem 1}\label{Appdendix_A}
\begin{proof}
	It is clear to see that constraints (19) and (29) are convex sets and $u_{m}$ is a continuous function. Then, we take the second-order derivative of $u_{m}$ with respect to $l_{m}$, which can be given by:
	\begin{equation}\label{}
		\textstyle
		\frac{\partial ^{2} u_{m}}{\partial l_{m}^{2}} = \frac{-2\sum l_{-m} (B+r({\bm{\upbeta}}_{m}))e^{-\lambda zI} 2^{-\pi}}{(l_{m} + \sum l_{-m})^3} \leq 0.
	\end{equation}
	
	Thus, $u_{m}$ is a concave function with respect to $l_{m}$, and problem (P3) is a convex optimization. There must exist an NE in the ERC game (see Theorem 3.2 in \cite{33}).
	%It is easy to show that $(P3)$ is a convex optimization, as the utility $u_m$ is a concave function of $L_m$ and the feasible set given by constraint (14) is a convex set. Thus, there must exist an NE in the ERC game (see Theorem 3.2 in \cite{33}).
\end{proof}

\section{Proof for Theorem 2}\label{Appdendix_B}
\begin{proof}
		%	Based on the best response strategy of $C_{m}$, we perform
	For $l_m^* \in [0, \bar l_m]$, we use some simple algebraic transformations on the best response strategy in (\ref{best response}). Then we have:
	\begin{equation}\label{NE1}
		\textstyle
		\sum_{j=1}^{M} l_{j}^{+} = \sqrt{\frac{\theta_{m}(\sum_{j=1}^{M} l_{j}^{+} - l_{m}^{+})}{p}},
	\end{equation}
	
	Let $\psi = \sum_{j=1}^{M} l_{j}^{+}$. For all coalitions, we have:
	\begin{equation}\label{}
		\textstyle
		\begin{split}
			\left\{
			\begin{aligned}
				&l_{1}^{+} = \psi - \frac{p \psi^{2}}{\theta_{1}}, ~0 \leq l_{1}^{+} \leq \bar l_m,\\
				&l_{2}^{+} = \psi - \frac{p \psi^{2}}{\theta_{2}}, ~0 \leq l_{2}^{+} \leq \bar l_m,\\
				&~~~~~\vdots      \\
				&l_{M}^{+} = \psi - \frac{p \psi^{2}}{\theta_{M}}, ~0 \leq l_{M}^{+} \leq \bar l_m.
			\end{aligned}
			\right.
		\end{split}	
		%		l_{m}^{NE} = \frac{M-1}{p \sum_{j=1}^{M} \frac{1}{\theta_{j}}}\cdot \left[ (1 - \frac{M-1}{\theta_{m} \sum_{j=1}^{M} \frac{1}{\theta_{j}}}) \right]^+, ~~ \forall \mathcal{C}_{m} \in \mathcal{C},
	\end{equation}
	
	Summing the above equations for all coalitions, we have:
	%	\begin{equation}\label{NE1}
		%		\textstyle
		%		l_{m}^{+} = \psi - \frac{p \psi^{2}}{\theta_{m}}, \quad \forall \mathcal{C}_{m} \in \mathcal{C},
		%	\end{equation}

	\begin{equation}\label{NE2}
		\textstyle
		\psi = \frac{M-1}{p \sum_{j=1}^{M} \frac{1}{\theta_{j}}}, ~ 0 \leq \psi \leq \sum_{m=1}^{M} \bar l_m.
	\end{equation}
	
	By substituting (\ref{NE2}) into (\ref{NE1}), we can obtain (\ref{NE-con}). Since $l_{m}^{NE}$ is an integer variable, we consider the rounded up and rounded down integer solutions of $l_{m}^{+}$, respectively. Then, we compare the utilities of two integer solutions and select the one with larger utility as $l_{m}^{NE}$. For the cases of $l_m^* < 0$ and $l_m^* > \bar l_m $,  $l_{m}^{NE}$ is bounded by $0$ and $\bar l_m$ by (\ref{best response}), respectively.
\end{proof}

\section{Proof for Lemma 1}\label{Appdendix_C}
\begin{proof}
	It can be observed from (\ref{um_lm}) and (\ref{best response}) that $u_{m}(l_{m})$ is continuous and differentiable in the domain $[0, \bar{l}_{m}]$.
	According to the Lagrange Mean Value Theorem \cite{add-2}, for any $l_{m}^{+}, ~l_{m}^{NE} \in [0, \bar{l}_{m}]$, there exists a $\zeta$ between $l_{m}^{+}$ and $l_{m}^{NE}$, such that:
	\begin{equation}\label{}
		\textstyle
		|u_{m}(l_{m}^{+}) - u_{m}(l_{m}^{NE})| = |\nabla u_{m}(\zeta)||l_{m}^{+} - l_{m}^{NE}|.
	\end{equation}
	
	It can be also observed from (\ref{best response}) that $u_{m}(l_{m})$ is gradient bounded in the domain $[0, \bar{l}_{m}]$, and then we have:
	\begin{equation}\label{}
		\textstyle
		|\nabla u_{m}(\zeta)| \leq G,
	\end{equation}
	where $G = \max_{l_{m}}|\nabla u_{m}(l_{m})|$ is the gradient upper bound constant.
	Since the maximum distance between $l_{m}^{+}$ and $l_{m}^{NE}$ is 1, and then we have:
	\begin{equation}\label{}
		\textstyle
		|u_{m}(l_{m}^{+}) - u_{m}(l_{m}^{NE})| \leq G |l_{m}^{+} - l_{m}^{NE}| \leq G.
	\end{equation}
\end{proof}

\section{Proof for Theorem 3}\label{Appdendix_D}
\begin{proof}
	Given any initial coalition structure $\mathcal{C^{\mathrm{0}}}$, MUs choose the atomic strategies that lead to changes in the coalition structure, represented by $\mathcal{C^{\mathrm{0}}} \to \mathcal{C^{\mathrm{1}}} \to \cdots \to \mathcal{C}^{s} \to \cdots$. Each coalition structure is a non-empty subset of $\mathcal{N}$:
	\begin{equation}\label{equation:coalition structure}
		\textstyle
		\mathcal{C}^{\mathit{s}} \subseteq \mathcal{X}(\mathcal{N}), ~\forall s,
	\end{equation}
	where $\mathcal{X}(\mathcal{N})$ is the set of all possible coalition structures for $\mathcal{N}$. Meanwhile, it can be seen that the new coalition structure contains the old coalition structure according to Definitions \ref{definition5}, \ref{definition6}, \ref{definition7}, \ref{definition8}, and \ref{definition9}. Their relationship can be expressed as
	\begin{equation}\label{equation:coalition relationship}
		\textstyle
		\mathcal{C^{\mathrm{0}}} \subset  \mathcal{C^{\mathrm{1}}} \subset  \cdots \subset  \mathcal{C}^{s} \subset  \cdots \subset \mathcal{C}^{S}.
	\end{equation}
	
	From Eq. (\ref{equation:coalition structure}) and Eq. (\ref{equation:coalition relationship}), we can see that each time MU chooses the atomic strategies, it results in a previously unvisited coalition structure. Based on the known fact that the partitions of a coalition structure are finite, as given by the Bell number \cite{35}, the number of changes in the coalition structure is finite. Thus, arbitrary initial coalition structures will converge to a stable coalition structure $\mathcal{C^{\mathit{S}}}$.
\end{proof}

%%%%%%%%%%%%%%%%%%%%%%%%%%%%%%%%%%% REFERENCE %%%%%%%%%%%%%%%%%%%%%%%%%%%%%%

\begin{IEEEbiography}[{\includegraphics[width=1in,height=1.25in,clip,keepaspectratio]{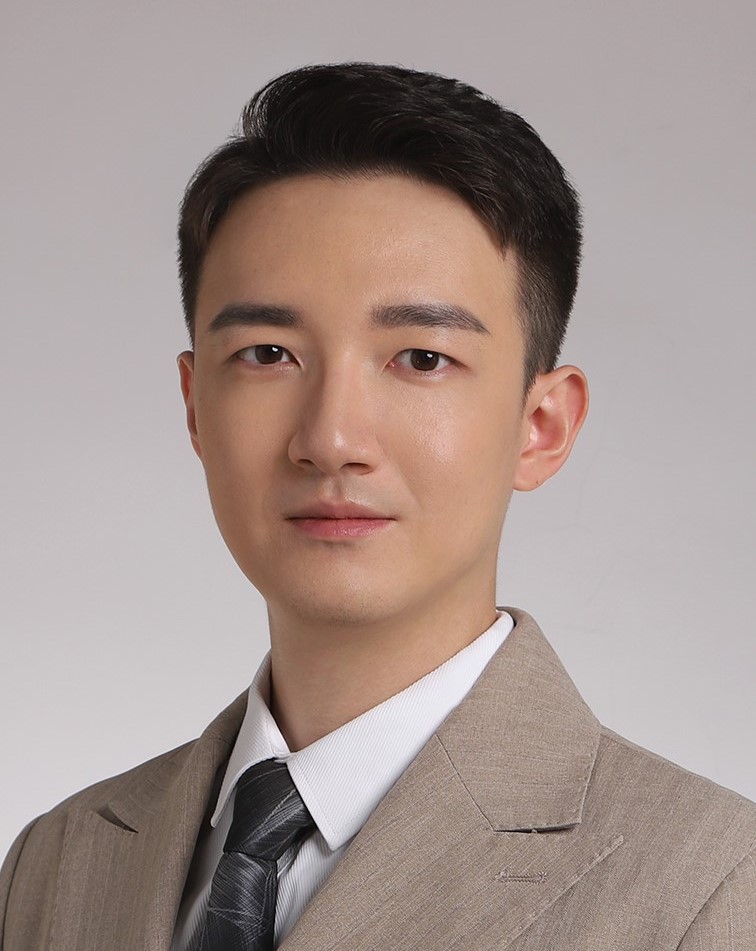}}]{Licheng Ye} (Graduate Student Member, IEEE) is currently working toward the Ph.D. degree with the School of Electronics and Information Engineering, Harbin Institute of Technology, Shenzhen, China. He is also a visiting Ph.D. student in Singapore University of Technology and Design, Singapore.
He received the M.S. degree from Harbin Institute of Technology, China, in 2022.
His research interests include mobile edge computing, blockchain, and federated learning.
\end{IEEEbiography}

\begin{IEEEbiography}[{\includegraphics[width=1in,height=1.25in,clip,keepaspectratio]{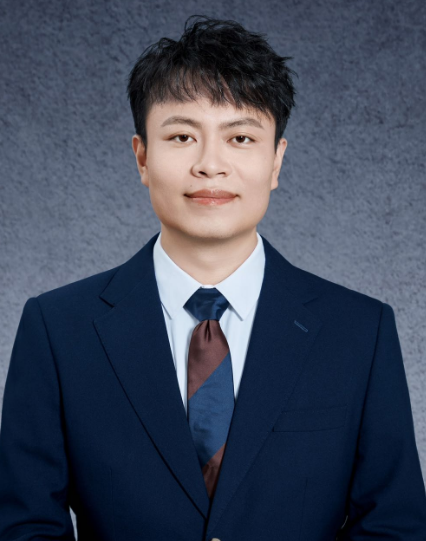}}]{Zehui Xiong} (Senior Member, IEEE) is currently a Full Professor with the School of Electronics, Electrical Engineering and Computer Science, Queen's University Belfast, United Kingdom. Prior to that, he was with Singapore University of Technology and Design, and Nanyang Technological University (NTU). He received his Ph.D. degree from NTU and was a visiting scholar with Princeton University and University of Waterloo.Recognized as a Clarivate Highly Cited Researcher, he has published over 250 peer-reviewed research papers in leading journals, with numerous Best Paper Awards from international flagship conferences. Featured in Forbes Asia 30U30, he serves as the Editor for many leading journals and Chair for numerous international conferences.His honors include the IEEE Asia Pacific Outstanding Young Researcher Award, IEEE VTS Early Career Award, IEEE Early Career Award for Excellence in Scalable Computing, IEEE Technical Committee on Blockchain and Distributed Ledger Technologies Early Career Award, IEEE Internet Technical Committee Early Achievement Award, IEEE TCSVC Rising Star Award, IEEE TCI Rising Star Award, IEEE TCCLD Rising Star Award, IEEE ComSoc Outstanding Paper Award, IEEE Best Land Transport Paper Award, IEEE Asia Pacific Outstanding Paper Award, IEEE CSIM Technical Committee Best Journal Paper Award, IEEE SPCC Technical Committee Best Paper Award, and IEEE Big Data Best Influential Conference Paper Award.
\end{IEEEbiography}

\begin{IEEEbiography}[{\includegraphics[width=1in,height=1.25in,clip,keepaspectratio]{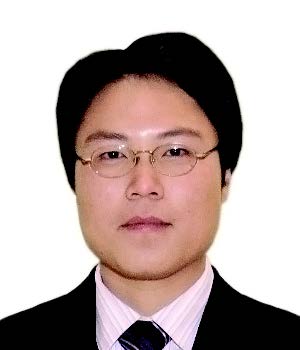}}]{Lin Gao} (Senior Member, IEEE) is a Professor at the School of Electronics and Information Engineering, Harbin Institute of Technology, Shenzhen, China. He received the Ph.D. degree in Electronic Engineering from Shanghai Jiao Tong University, Shanghai, China, in 2010. His main research interests are in the interdisciplinary area between game theory, optimization, and machine learning, with particular focuses on reinforcement learning, federated learning, crowd/edge intelligence, mobile edge computing,   cognitive communication and networking. He is the co-recipient of 5 Best Paper Awards from leading conference proceedings on wireless communications and networking. He received the IEEE ComSoc Asia-Pacific Outstanding Young Researcher Award in 2016.
\end{IEEEbiography}

\begin{IEEEbiography}[{\includegraphics[width=1in,height=1.25in,clip,keepaspectratio]{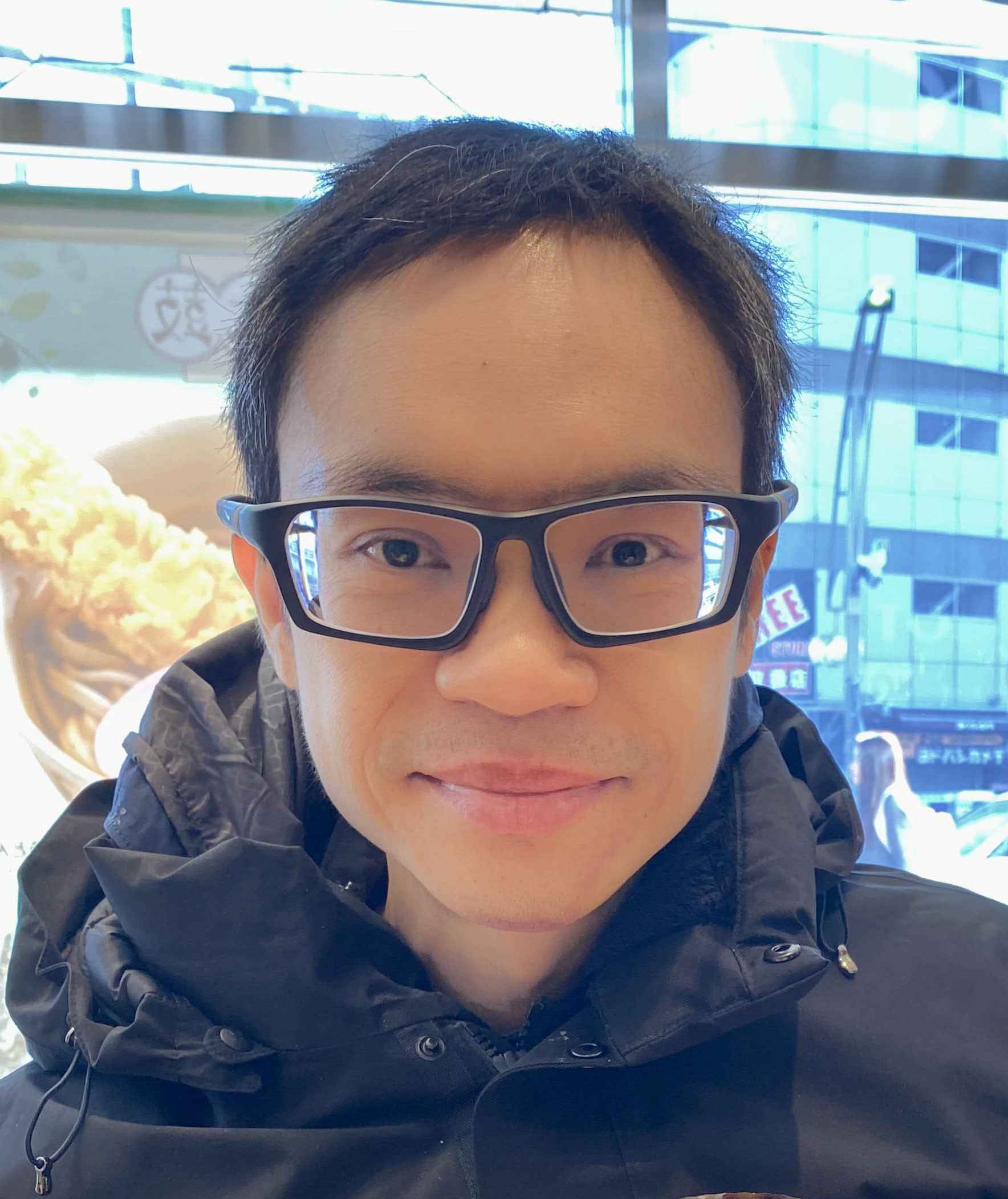}}]{Dusit Niyato} (M'09-SM'15-F'17) is a professor in the College of Computing and Data Science, at Nanyang Technological University, Singapore. He received B.Eng. from King Mongkuts Institute of Technology Ladkrabang (KMITL), Thailand and Ph.D. in Electrical and Computer Engineering from the University of Manitoba, Canada. His research interests are in the areas of mobile generative AI, edge intelligence, quantum computing and networking, and incentive mechanism design.
\end{IEEEbiography}

\end{document}